\documentclass[final,5p,times,twocolumn]{elsarticle}
\usepackage{amsmath,amssymb,amsthm}
\usepackage{graphicx}
\usepackage{physics}
\usepackage{optidef}
\usepackage{graphicx,epstopdf,epsfig}
\usepackage{adjustbox}
\usepackage[dvipsnames]{xcolor}
\usepackage{float}
\usepackage[caption=false,font=footnotesize]{subfig}
\usepackage{siunitx}
\usepackage{url}
\usepackage[normalem]{ulem}
\usepackage{booktabs}
\usepackage{centernot}
\usepackage{enumitem}
\usepackage{textcomp}
\usepackage{hhline}
\usepackage{stackengine}

\newcommand{\Sy}{\mathbb{S}}
\renewcommand{\Re}{\mathbb{R}}

\newcommand{\Co}{\mathbb{C}}

\newcommand{\Ur}{\mathbb{U}}

\newcommand{\cN}{\mathcal{N}}

\newcommand{\cR}{\mathcal{R}}

\def\A{\mathbf{A}}

\def\D{\mathbf{D}}
\def\E{\mathbf{E}}
\def\F{\mathbf{F}}

\def\I{\mathbf{I}}
\def\K{\mathbf{K}}

\def\S{\mathbf{S}}

\def\U{\mathbf{U}}
\def\V{\mathbf{V}}
\def\W{\mathbf{W}}

\def\ZE{\mathbf{0}}
\def\LA{\mathbf{\Lambda}}

\def\a{\boldsymbol{a}}
\def\b{\boldsymbol{b}}
\def\c{\boldsymbol{c}}

\def\q{\boldsymbol{q}}
\def\u{\boldsymbol{u}}
\def\v{\boldsymbol{v}}

\def\x{\boldsymbol{x}}
\def\y{\boldsymbol{y}}
\def\z{\boldsymbol{z}}

\def\ze{\mathbf{0}}
\def\one{\mathbf{1}}

\def\boldeta{\boldsymbol{\eta}}

\def\bxi{\boldsymbol{\xi}}

\DeclareMathOperator*{\argmin}{\arg\!\min}
\DeclareMathOperator*{\prox}{\mathrm{prox}}
\DeclareMathOperator*{\minimize}{minimize}

\DeclareMathOperator*{\dist}{dist}
\DeclareMathOperator*{\Span}{span}

\DeclareMathOperator*{\Prob}{\mathrm{Prob}}

\DeclareMathOperator*{\Img}{\mathcal{R}}

\graphicspath{{./Images/}}
\allowdisplaybreaks

\newtheorem{theorem}{Theorem}[section]
\newtheorem{proposition}[theorem]{Proposition}

\newtheorem{corollary}[theorem]{Corollary}
\newtheorem{lemma}[theorem]{Lemma}

\newtheorem{remark}[theorem]{Remark}

\begin{document}

\begin{frontmatter}
\title{On Exact and Robust Recovery for Plug-and-Play Compressed Sensing\tnoteref{label1}}
\tnotetext[label1]{C. D. Athalye and K. N. Chaudhury are with the Department of Electrical Engineering, Indian Institute of Science, Bengaluru 560012, India.
At the time this research was carried out, R. G. Gavaskar was with the Department of Electrical Engineering, Indian Institute of Science, Bengaluru 560012, India. He is now with Qualcomm India.
}

\author{Ruturaj G. Gavaskar\corref{corresponding}}
\cortext[corresponding]{Corresponding author.}
\ead{ruturajg@iisc.ac.in}
\author{Chirayu D. Athalye}
\ead{chirayu@iisc.ac.in}
\author{Kunal N. Chaudhury}
\ead{kunal@iisc.ac.in}

\affiliation{organization={Department of Electrical Engineering, Indian Institute of Science},
            city={Bangalore},
            postcode={560012}, 
            country={India}}

\begin{abstract}
In the Plug-and-Play (PnP) framework, regularization is performed by plugging an off-the-shelf denoiser within a proximal algorithm such as ISTA or ADMM.
PnP produces state-of-the-art results in many imaging applications, but its theoretical aspects are not well understood. 
In particular, the present work is motivated by the question that, similar to classical compressed sensing, is it theoretically possible to recover the ground-truth using PnP?
More specifically, under what conditions on the ground-truth, the sensing matrix, and the PnP denoiser is the reconstruction guaranteed to be exact?
The foremost hurdle in this regard is the absence of an explicit regularizer\,--\,PnP is an algorithmic framework, and it is not apparent if a limit point of the PnP iterations (if one exists) is the minimizer of some objective function.
It was recently shown that it is possible to associate a convex regularizer $\Phi$ with a class of linear denoisers.
For such denoisers, the PnP iterations correspond to solving a convex optimization problem involving $\Phi$.
Motivated by this result, we consider the PnP analogue of the compressed sensing problem: $\mathrm{\min \, \Phi(\x) \, s.t. \, \A\x=\A\bxi}$, where $\A \in \Re^{m \times n}$ is a random sensing matrix, $\Phi$ is the regularizer associated with a denoiser $\W$ from the class mentioned above, and $\bxi$ is the ground-truth signal.
We prove that if the sensing matrix is Gaussian and $\bxi \in \mathrm{range}(\W)$, then the minimizer of this problem is almost surely $\bxi$  if  $\mathrm{rank}(\W) \leqslant m$, and almost never if $\mathrm{rank}(\W) > m$. 
In other words, the range of the PnP denoiser plays the role of a signal prior, and its dimension marks a sharp transition from failure to success of exact recovery. 
We are able to extend the result to subgaussian sensing matrices, except that we can guarantee exact recovery only with high probability (and not almost surely).
For noisy measurements of the form $\b=\A\bxi+\boldeta$, we consider a robust formulation: $\mathrm{\min \, \Phi(\x) \, s.t. \, \lVert \A\x -\b \rVert \leqslant \delta}$. 
We prove that if $\x^\ast$ is an optimal solution of this problem, then with high probability, the distortion $\lVert \x^\ast - \bxi\rVert$ can be bounded by $\lVert \boldeta \rVert$ and $\delta$, provided the number of measurements $m$ is sufficiently large. 
In particular, we can derive the sample complexity of compressed sensing as a function of distortion error and success rate. We discuss the extension of these results to random Fourier measurements.
To the best of our knowledge, this is the first work that gives probabilistic recovery guarantees for compressed sensing using PnP regularization.
We perform numerical experiments to validate our theoretical findings and discuss research directions stemming from this work.
\end{abstract}


\end{frontmatter}

\section{Introduction}
\label{sec:intro}

Linear inverse problems such as deblurring, superresolution, and compressed sensing come up in image recovery applications from partial or corrupted measurements \cite{Mallat1999_wavelet_tour,Candes2008_compressive_sampling}.
The abstract problem is that we are given measurements $\b \in \Re^m$ of the form 
\begin{equation}
\label{eq:forward-model}
\b = \A \bxi + \boldeta,
\end{equation}
where $\bxi \in \Re^n$ is the ground-truth image, $\boldeta$ is white Gaussian noise, and $\A \in \Re^{m \times n}$ is the application-specific forward model.
The objective is to recover $\bxi$ from $\b$ and $\A$.
This problem is ill-posed as stated; hence the need for regularization \cite{Hunt1977_bayesian_methods}.
The standard approach is to pose the recovery task as an optimization problem: 
\begin{equation}
\label{eq:regularized-lsq}
\minimize_{\x \in \Re^n} \, f(\x) + \lambda \Phi(\x),
\end{equation}
where $f(\x) := \norm{\A \x - \b}^2 /2$ is the loss function, $\Phi \colon \Re^n \to \Re$ is some regularizer and $\lambda > 0$ is a tuning parameter. Here and henceforth, $\norm{\cdot}$ will denote the Euclidean norm. If $\Phi$ is convex, \eqref{eq:regularized-lsq} can be solved using iterative algorithms such as ISTA and ADMM \cite{Beck2017_optimization}.
These algorithms require that the \textit{proximal map} of $\Phi$, 
\vspace{-1mm}
\begin{equation}
\label{eq:prox}
\prox_{\Phi}(\u) = \argmin_{\x \in \Re^n}  \, \frac{1}{2} \norm{\x - \u}^2 + \Phi(\x),
\vspace{-1mm}
\end{equation}
can be computed efficiently (in closed form or iteratively).
For example, the ISTA update $\x_k \to \x_{k+1}$ is given by
\begin{equation*}
\x_{k+1} = \prox_{\tau \Phi} \big( \x_k - \tau\nabla \! f (\x_k)\big),
\end{equation*}
where $\tau > 0$ is a constant step size. From a Bayesian viewpoint, $\prox_{\Phi}(\u)$ performs denoising of $\u$ where the prior on the ground-truth $\x$ is derived from $\Phi$  \cite{Hunt1977_bayesian_methods}. 
Motivated by this observation,  \textit{Plug-and-Play} (PnP) regularization was proposed in \cite{Venkatakrishnan2013_PnP,Sreehari2016_PnP}.
In PnP, the proximal map within ISTA or ADMM is replaced by a powerful Gaussian denoiser $D : \Re^n \to \Re^n$, such as NLM \cite{Buades2005_NLM}, BM3D \cite{Dabov2007_BM3D}, etc.
For example, applied to ISTA, the PnP update $\x_k \to \x_{k+1}$ becomes
\begin{equation}
\label{pnp-ista}
\x_{k+1} = D\big( \x_k - \tau \nabla f(\x_k) \big).
\end{equation}
The updates for ADMM are more involved than ISTA, and we refer the reader to \cite{Beck2017_optimization} for details.

The core idea in PnP is to directly deploy the denoiser instead of having to specify $\Phi$ and go through its proximal map.
Although this is somewhat ad hoc, remarkably, PnP has been shown to work well in practice for many imaging applications \cite{Sreehari2016_PnP,Ahmad2020_PnP_MRI,Yuan2020_PnP_snapshot_CS,Zhang2021_PnP_deep_prior}.
Following the empirical success of PnP, its theoretical aspects have been investigated in several works; see for example \cite{Ryu2019_PnP_trained_conv,Liu2021_PnP_REC} and references therein.
A fundamental question is, can PnP be interpreted as a regularization mechanism? 
This translates to whether the denoiser in PnP can be expressed as the proximal map of a (convex) function.
This is unlikely to be true for nonlinear denoisers such as DnCNN and BM3D.
On the other hand, it is shown in a series of papers that an explicit convex regularizer $\Phi$ can be associated with specific \textit{linear} denoisers \cite{Sreehari2016_PnP,Teodoro2019_PnP_fusion,Chan2019_PnP_graph_SP,Gavaskar2021_PnP_linear_denoisers,Nair2021_PnP_fixed_point}.
In particular, the following is a restatement of \cite[Theorem 2]{Teodoro2019_PnP_fusion}.
\begin{theorem}
\label{thm:teodoro}
Let $D$ be a linear operator of the form $D(\x) = \W \x$, where $\W \in \Re^{n \times n}$ is symmetric and has eigenvalues in $[0,1]$.
Then $D$ is the proximal map of the following (extended-real-valued) convex function:
\begin{equation}
\label{eq:kernel-regularizer}
\Phi_{\W}(\x) :=
\begin{cases}
\frac{1}{2} \x^\top (\I - \W) \W^\dagger \x, & \mathrm{if} \ \x \in \cR(\W),\\
+\infty, & \mathrm{otherwise},
\end{cases}
\vspace{-1mm}
\end{equation}
\end{theorem}
where $\cR(\W)$ denotes the range of $\W$ and $\W^\dagger$ is the pseudoinverse of $\W$.

We note that the expression of $\Phi_{\W}$ in \cite{Teodoro2019_PnP_fusion}, which is given using a condensed eigenvalue decomposition of $\W$, can be shown to be equivalent to \eqref{eq:kernel-regularizer}.
By Theorem \ref{thm:teodoro}, we can associate the regularizer $\Phi_{\W}$ with $\W$. 
Subsequently, it is not difficult to  establish convergence of PnP \cite{Sreehari2016_PnP,Afonso2010_CSALSA,Gavaskar2021_PnP_linear_denoisers}. For example, if $(\x_{k})$ is the sequence generated by \eqref{pnp-ista}, where $D$ is the linear denoiser $\W$, then $f(\x_k) + \lambda \Phi_{\W}(\x_k)$ converges to the minimum of $f + \lambda \Phi_{\W}$ \cite{Nair2021_PnP_fixed_point}. Practical denoisers satisfying the condition in Theorem \ref{thm:teodoro} include DSG-NLM \cite{Sreehari2016_PnP}, GMM \cite{Teodoro2019_PnP_fusion} and GLIDE \cite{Talebi2013_GLIDE}.

The next natural question is how strong is the prior induced by $\Phi_{\W}$, i.e., how well can it capture the characteristics of the ground-truth image?
We turn to the theory of compressed sensing (CS) to answer this question.
A classical result in CS theory states that if $\A$ is a random Gaussian matrix and if the ground-truth $\bxi$ is sparse, then $\bxi$ can be recovered approximately with high probability using $\ell_1$ minimization \cite{Candes2008_compressive_sampling,Candes2006_stable_recovery}, i.e., by solving the problem
\begin{mini*}
{}{\|\x\|_1}{}{}
\addConstraint{\|\A \x - \b\|}{\leqslant \delta.}
\end{mini*}
In particular, if $\boldeta = \ZE$ (clean measurements), then with high probability $\bxi$ can be recovered exactly by solving
\begin{mini}
{}{\|\x\|_1}{\label{eq:CS_exact}}{}
\addConstraint{\A \x}{= \A \bxi.}
\end{mini}
In this work, we explore whether similar guarantees can be obtained for compressed sensing using PnP.
More specifically, we ask the following questions.
\begin{enumerate}[label=(\roman*)]
\item \textbf{Exact Recovery}: Consider the analogue of \eqref{eq:CS_exact} using the PnP regularizer $\Phi_{\W}$:
\begin{mini*}
{}{\Phi_{\W}(\x)}{}{}
\addConstraint{\A \x }{= \A \bxi.}
\end{mini*}
Using \eqref{eq:kernel-regularizer}, we can rewrite the above problem as follows:
\begin{equation}
\tag{$\text{P}_0$}
\label{eq:noiseless_prob}
\begin{aligned}
&\minimize && \x^\top (\I-  \W)\W^\dagger \x \\
&\mathrm{subject \; to} && \A \x = \A \bxi, \quad \x \in \cR(\W).
\end{aligned}
\end{equation}
Is $\bxi$ the unique minimizer of this problem?
Since we work with a random $\A$, any such guarantee will be probabilistic.

\item \textbf{Robust Recovery}: Consider the general problem of recovery in the presence of measurement noise:
\begin{equation}
\tag{$\text{P}_{\delta}$}
\label{eq:noisy_prob}
\begin{aligned}
&\minimize && \x^\top (\I-  \W)\W^\dagger \x \\
&\mathrm{subject \; to} && \norm{\A \x - \b}^2 \leqslant \delta^2, \quad \x \in \cR(\W).
\end{aligned}
\end{equation}
Let  $\x^\ast$ be a minimizer of \eqref{eq:noisy_prob}.
If $\boldeta$ is small, can we guarantee that the error $\norm{\x^\ast - \bxi}$ is small?
Furthermore, can we bound $\norm{\x^\ast - \bxi}$ in terms of $\boldeta$?
\end{enumerate}

A natural question is why do we switch from the unconstrained problem \eqref{eq:regularized-lsq} to the constrained formulations \eqref{eq:noiseless_prob} and \eqref{eq:noisy_prob}?
The reason is that the hard constraints in \eqref{eq:noiseless_prob} and \eqref{eq:noisy_prob} make the theoretical analysis more tractable than the unconstrained problem \eqref{eq:regularized-lsq}, where the regularizer $\Phi_{\W}$ imposes only a soft penalty.
This is indeed inspired from the classical CS theory \cite{Candes2006_stable_recovery,Candes2008_compressive_sampling}, where the same trick is used for simplifying the analysis.
Moreover, it is evident that the exact recovery is improbable to achieve in the unconstrained case (even in the absence of noise) since $\Phi_{\W}$ is a smoothly varying function. 
Subsequently, investigating a possibility of exact recovery necessitates switching to the constrained formulation. 
We also note that \eqref{eq:regularized-lsq} and \eqref{eq:noisy_prob} are equivalent for appropriate choices of $\delta$ and $\lambda$ \cite{Figueiredo2007_gradient_proj}.
On the algorithmic side, the question is whether \eqref{eq:noisy_prob} can be solved as in classical PnP \cite{Sreehari2016_PnP}, namely, by plugging denoiser $D$ into some suitable proximal algorithm? 
As shown in \cite{Unni2022_CSALSA_fusion}, this can indeed be done within the framework of the ADMM algorithm.

We expect that exploring the above questions will help us understand why PnP works well in practice.
Recent works such as \cite{Ahmad2020_PnP_MRI,Liu2021_PnP_REC} have successfully used PnP for reconstructing images from compressively sensed measurements, albeit using nonlinear denoisers.
Linear symmetric denoisers of the form in Theorem \ref{thm:teodoro} are well suited to explore questions in this area since they induce a convex regularizer that can be expressed using an explicit formula.
In fact, this property is known to be true so far only for linear denoisers \cite{Gavaskar2021_PnP_linear_denoisers,Nair2021_PnP_fixed_point}.
The linearity of $\W$ coupled with the convexity of $\Phi_{\W}$ in Theorem \ref{thm:teodoro} makes the problem tractable.

In this paper, we provide probabilistic guarantees on exact and robust compressed sensing recovery that address the questions posed above.
We focus on the case where $\A$ is a random Gaussian or Rademacher matrix.
In the Gaussian case, we prove that it is improbable to achieve exact recovery if the rank of $\W$ is greater than $m$ (Theorem \ref{thm:improbability}).
This leads us to consider low-rank denoisers such as the GLIDE filter \cite{Talebi2013_GLIDE}.
We prove that for low-rank denoisers, exact recovery is achieved with probability $1$ if $\A$ is Gaussian and $m \geqslant \rank(\W)$ (Theorem \ref{thm:exact_rec_prob_gaussian_r<m}), and with high probability if $\A$ is Rademacher and $m \geqslant O\big( \rank(\W) \big)$ (Theorem \ref{thm:exact_rec_prob_r<m}).
Furthermore, we prove that robust recovery is possible with high probability for both Gaussian and Rademacher $\A$ (Theorem \ref{thm:robust_rec_prob_r<m}). In particular, we obtain the sample complexity of robust compressed sensing as a function of distortion error and success rate.
We briefly discuss a possible extension of our results to randomized sensing matrices in bounded orthonormal systems, such as discrete Fourier or Hadamard projections (Section \ref{subsec:Fourier-Hadamard}). 
Our analysis is inspired from the classical CS theory; therefore, most of our probabilistic guarantees bear resemblance to analogous classical CS results on exact and robust recovery.
To the best of our knowledge, this is the first work to provide a connection between classical CS results and CS using PnP.
We note that a preliminary version of this work appears in a conference proceeding \cite{Gavaskar2022_robust_recovery}, where the focus is mainly on empirical observations. 

Throughout this paper, unless specified otherwise, $\W$ denotes a $(n \times n)$ symmetric matrix with eigenvalues in $[0,1]$, i.e., $\W$ is a linear denoiser satisfying the conditions in Theorem \ref{thm:teodoro}. 
We consider symmetric denoisers in this paper just to keep the exposition simple.
Note that all our results can be extended to the case where $\W$ is a non-symmetric denoiser such as a kernel filter \cite{Milanfar2013_filtering_tour}; see discussion in Section \ref{sec:closing_remarks} in this regard. 

We state and discuss the main results in Section \ref{sec:main}. Proofs of these results are deferred to Section \ref{sec:proofs} and their
numerical validation to Section \ref{sec:numerical}.
In Section \ref{sec:discussion}, we relate our work to existing compressed sensing literature, as well as discuss some implications and future research directions arising from our work.

\section{Exact and Robust Recovery}
\label{sec:main}

In this section, we formally state and discuss our results on exact and robust recovery.
The technical proofs are deferred to Section \ref{sec:proofs}.

We first focus on the case where $\A$ is a $(m \times n)$ random sensing matrix.
In particular, if the entries of $\A$ are i.i.d. Gaussian with mean $0$ and variance $1/m$, then we refer to $\A$ as a \textit{random Gaussian matrix} \cite{Vershynin2018_HDP}.
First, we state an improbability result which implies that exact recovery is improbable from random Gaussian measurements if the rank of $\W$ is greater than the number of measurements unless $\bxi$ is a fixed point of $\W$.
\begin{theorem}
\label{thm:improbability}
Let $\A$ be a random Gaussian matrix and $\W$ be (statistically) independent of $\A$. 
Let $\bxi$ be a feasible point of \eqref{eq:noiseless_prob}, i.e., $\bxi \in \cR(\W)$. 
If $m < \rank(\W)$ and $\bxi \notin \cN(\I - \W)$, then with probability $1$, $\bxi$ is \emph{not} a minimizer of \eqref{eq:noiseless_prob}.
\end{theorem}

The theorem can be interpreted as follows: if the rank of $\W$ is large, then unless $\bxi \in \cN(\I - \W)$, we can never recover $\bxi$ even if $\bxi$ is a feasible point of \eqref{eq:noiseless_prob}.
An interesting case is when $\W$ is doubly-stochastic and irreducible; for example, DSG-NLM \cite[Appendix B]{Sreehari2016_PnP}.
For such matrices, the Perron-Frobenius theorem implies that $\cN(\I - \W)=\{\alpha \one : \alpha \in \Re\}$ \cite{Milanfar2013_filtering_tour}.
It follows from Theorem \ref{thm:improbability} that we can almost never achieve exact recovery except for the uninteresting case where $\bxi = \alpha \one$ for some $\alpha \in \Re$; i.e., $\bxi$ is a constant signal.

\begin{corollary}
\label{cor:improbability_kernel}
If $\A$ is a random Gaussian matrix, $\W$ is doubly-stochastic and irreducible, and $\bxi$ is not a constant signal, then the probability that $\bxi$ is a solution of \eqref{eq:noiseless_prob} is $0$.
\end{corollary}

In the light of Theorem \ref{thm:improbability}, we focus on low-rank denoisers. As an example, consider the GLIDE filter \cite{Talebi2013_GLIDE}. This is a symmetric denoiser that satisfies the conditions in Theorem \ref{thm:teodoro} and whose rank is user-configurable. Importantly, as discussed in \cite{Talebi2013_GLIDE}, GLIDE is able to achieve denoising quality comparable to NLM \cite{Buades2005_NLM} and BM3D \cite{Dabov2007_BM3D} while having rank in the low hundreds (say, $200$).

The following theorem states that if $\bxi$ is a feasible point of \eqref{eq:noiseless_prob}, then exact recovery can be achieved almost surely from sufficiently many random Gaussian measurements.
\begin{theorem}
\label{thm:exact_rec_prob_gaussian_r<m}
Let $\A$ be a $(m \times n)$ random Gaussian matrix, and $\W$ be (statistically) independent of $\A$.
If $m \geqslant \rank(\W)$ and $\bxi \in \cR(\W)$, then $\bxi$ is the unique minimizer of \eqref{eq:noiseless_prob} with probability $1$.
\end{theorem}
Theorem \ref{thm:exact_rec_prob_gaussian_r<m} implies that more measurements are required for exact recovery if $\W$ has a large rank.
Since the essence of compressed sensing is to work with fewer measurements, we should thus use a low-rank denoiser.
On the other hand, reducing the rank of $\W$ shrinks the space of recoverable signals since $\bxi$ is required to lie in $\cR(\W)$.
Thus, choosing the rank of the denoiser involves a trade-off between the number of measurements and the space of exactly recoverable signals.

The proof of Theorem \ref{thm:exact_rec_prob_gaussian_r<m} does not use any property of the Gaussian distribution other than absolute continuity (i.e., it admits a density function).
Therefore, Theorem \ref{thm:exact_rec_prob_gaussian_r<m} holds for any random sensing matrix $\A$ whose entries are independent continuous random variables.
Theorems \ref{thm:improbability} and \ref{thm:exact_rec_prob_gaussian_r<m} together imply that if $\bxi \in \cR(\W)$, then we obtain exact recovery from random Gaussian measurements with probability $1$ if $m \geqslant \rank(\W)$ and with probability $0$ if $m < \rank(\W)$.

Note that for random \textit{subgaussian} sensing matrices, Theorem \ref{thm:exact_rec_prob_gaussian_r<m} is not necessarily applicable because subgaussian random variables need not be continuous.
A random variable $X$ is said to be subgaussian if $\Prob\big[|X| > \alpha\big] \leqslant 2 e^{-c \alpha^2}$ for some constant $c > 0$ and all $\alpha > 0$ \cite{Vershynin2018_HDP}; examples of subgaussian random variables are given in \cite[Sec. 2.5]{Vershynin2018_HDP}.
Note that in particular, the Rademacher distribution,
\begin{equation*}
\Prob\big[X = 1\big] = \Prob\big[X = -1\big] = 1/2,
\end{equation*}
and the Gaussian distribution are both subgaussian.
An $m \times n$ random matrix $\A$ is said to be subgaussian if its entries are i.i.d. subgaussian random variables with mean $0$ and variance $1/m$.
However, in this paper, unless specified otherwise, we restrict the term ``subgaussian'' to specifically mean either Gaussian or Rademacher distributions.
The following property of subgaussian matrices can be found in \cite{Baraniuk2008_RIP_proof,Matouvsek2008_JL_variants}.
\begin{lemma}
\label{lem:gamma}
Let $\A$ be a $(m \times n)$ random subgaussian matrix.
Then there exists a function $\gamma : (0,1) \to \Re_+$ such that for any $\x \in \Re^n$ that is independent of $\A$,
\begin{equation*}
\label{eq:subgaussian_concentration}
\Prob\Big[ (1 - \epsilon) \norm{\x}^2 \leqslant \norm{\A \x}^2 \leqslant (1 + \epsilon) \norm{\x}^2 \Big] \geqslant 1 - 2 e^{-m \gamma(\epsilon)}
\end{equation*}
for all $\epsilon \in (0,1)$.
\end{lemma}

More specifically, $\gamma(\epsilon) := \epsilon^2 / 6$ if $\A$ is Gaussian \cite[Lemma 23.3]{Shwartz2014_UML}, and $\gamma(\epsilon) := \epsilon^2 / 4 - \epsilon^3 / 6$ if $\A$ is Rademacher \cite[Lemma 4]{Achlioptas2001_database_friendly_random_proj}. 
Thus, for random subgaussian matrices, $\gamma$ is continuous and strictly increasing.
The following theorem gives a probabilistic guarantee of exact recovery for subgaussian sensing matrices.
\begin{theorem}
\label{thm:exact_rec_prob_r<m}
Let $\A$ be a $(m \times n)$ random subgaussian matrix, $\W$ be independent of $\A$, and $\bxi \in \cR(\W)$.
For $\beta \in (0,1)$, suppose
\begin{equation}
\label{eq:m_bound_subgaussian}
m \geqslant \frac{\ln(2/\beta) + r \ln(12/0.99)}{\gamma(0.99/2)},
\end{equation}
where $r = \rank(\W)$ and $\gamma$ is the function in Lemma \ref{lem:gamma}.
Then with probability at least $1 - \beta$, $\bxi$ is the unique minimizer of \eqref{eq:noiseless_prob}.
\end{theorem}

For the robust recovery problem \eqref{eq:noisy_prob} with $\A$ as a random subgaussian matrix, the following theorem gives a probabilistic bound on $\|\x^\ast -\bxi\|$.
\begin{theorem}
\label{thm:robust_rec_prob_r<m}
Let $\A$ be a $(m \times n)$ random subgaussian matrix, and $\W$ be independent of $\A$.
Let $\Omega_\delta \neq \emptyset$ and $\x^\ast$ be a minimizer of \eqref{eq:noisy_prob}.
For $\epsilon, \beta \in (0,1)$, suppose
\begin{equation}
\label{eq:m_bound_robust}
m \geqslant \frac{\ln(4/\beta) + r \ln(12/\epsilon)}{\gamma(\epsilon/2)},
\end{equation}
where $r = \rank(\W)$ and $\gamma$ is the function in Lemma \ref{lem:gamma}.
Then with probability at least $1 - \beta$, 
\begin{equation}
\label{eq:robust_bound_r<m}
\norm{\x^\ast - \bxi} \leqslant \left( 1 + \frac{2}{1 - \epsilon} \right) \dist \big( \bxi,\cR(\W) \big) + \frac{\delta + \norm{\boldeta}}{1 - \epsilon},
\end{equation}
where $\dist \big( \bxi,\cR(\W) \big)$ is the distance of $\bxi$ from $\cR(\W)$.
\end{theorem}

For the Gaussian case in particular, the lower bound in \eqref{eq:m_bound_robust} reduces to $O\big(\epsilon^{-2} r \ln(1/\epsilon) \big)$.
Theorem \ref{thm:robust_rec_prob_r<m} involves a trade-off between the lower bound on $m$ and the upper bound on the recovery error.
For a fixed denoiser $\W$ and probability $1 - \beta$, the lower bound in \eqref{eq:m_bound_robust} decreases from $+\infty$ to a finite value as $\epsilon$ increases from $0$ to $1$; this is because $\gamma(\epsilon) \to 0$ as $\epsilon\to 0$ for both Gaussian and Rademacher matrices.
On the other hand, the upper bound in \eqref{eq:robust_bound_r<m} increases to $+\infty$ as $\epsilon$ increases from $0$ to $1$.
Note that $\beta$ can be interpreted as the failure rate of robust recovery, whereas $\epsilon$ is a parameter that controls the recovery accuracy.
According to Theorem \ref{thm:robust_rec_prob_r<m}, we need more measurements for accurate recovery with high success rate; this is consistent with intuition.

Since $n \geqslant m$, the condition given by \eqref{eq:m_bound_robust} is fulfilled provided the lower bound is at most $n$. 
In Appendix \ref{sec:minimum_n}, we explain that if $n > (\ln 4 + r \ln 12)/\gamma(1/2)$, then \eqref{eq:m_bound_robust} is satisfied for $(\beta,\epsilon)$ belonging to an appropriate subset of $(0,1) \times (0,1)$; see Proposition \ref{prp:lower-bound on m}. 
Subsequently, Theorem \ref{thm:robust_rec_prob_r<m} is applicable for large-sized signals such as images. A similar observation applies to Theorem \ref{thm:exact_rec_prob_r<m}.

\begin{remark}
Note that Lemma \ref{lem:gamma} holds for random subgaussian matrices which are neither Gaussian nor Rademacher, with the difference being that $\epsilon$ is allowed to take values in $(0,1/2]$; see \cite[Theorem 3.1]{Matouvsek2008_JL_variants}. 
Subsequently, Theorems \ref{thm:exact_rec_prob_r<m} and \ref{thm:robust_rec_prob_r<m} have counterparts for other types of subgaussian matrices.
\end{remark}

\begin{figure*}[t]
\centering
\subfloat[]{\raisebox{3.2mm}{\includegraphics[width=0.311\linewidth]{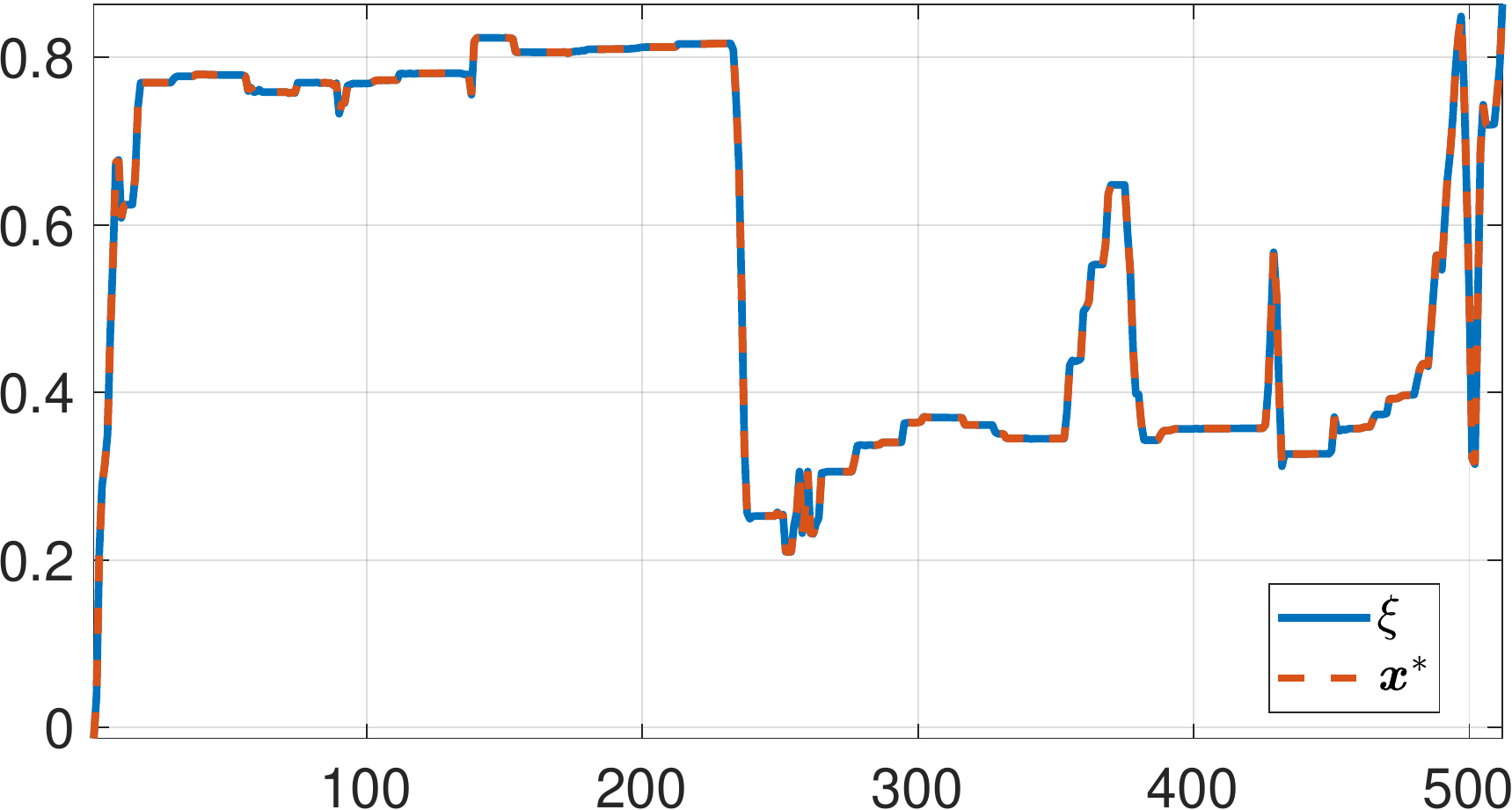}}}
\hspace{2mm}
\subfloat[]{\includegraphics[width=0.332\linewidth]{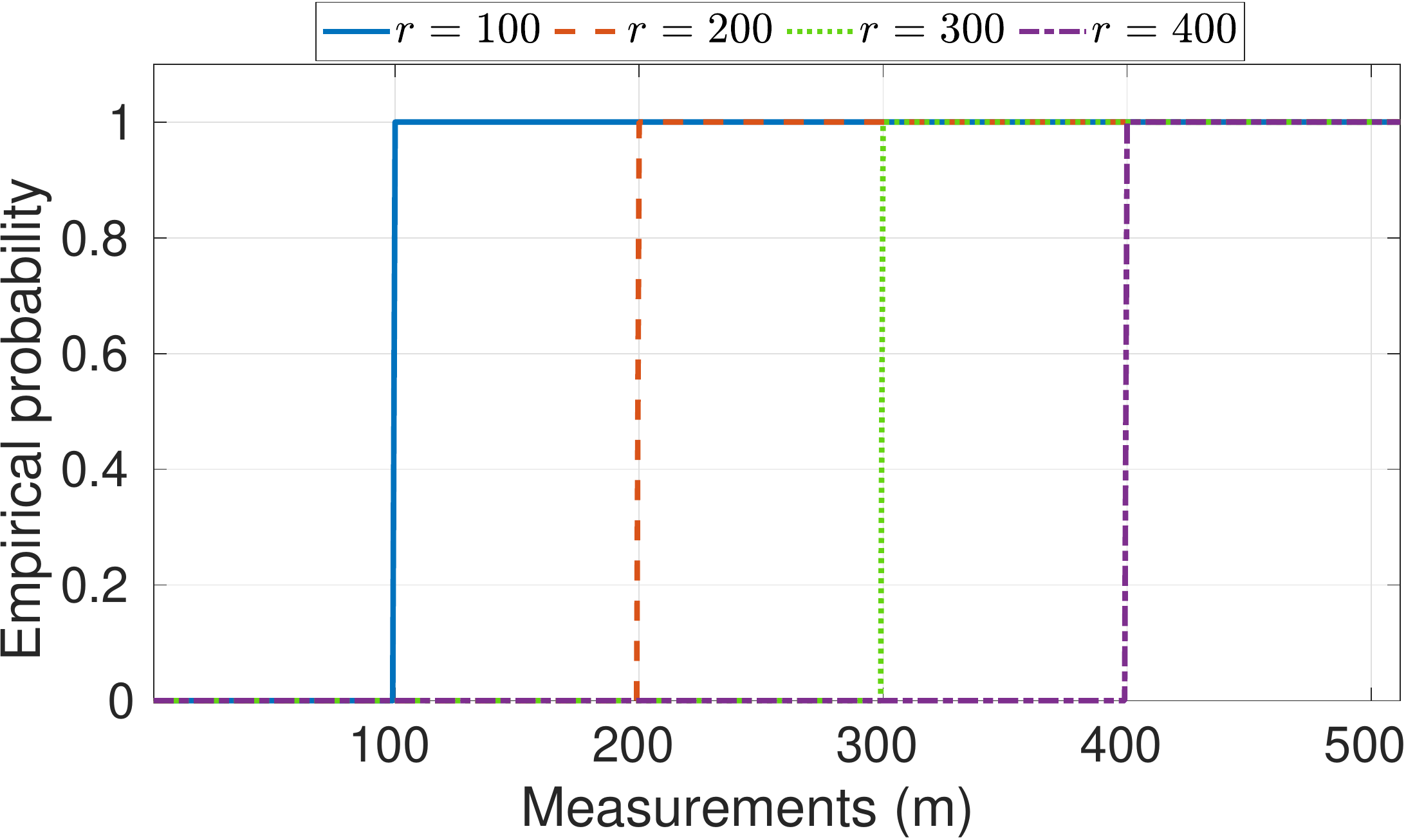}}
\hspace{2mm}
\subfloat[]{\includegraphics[width=0.24\linewidth]{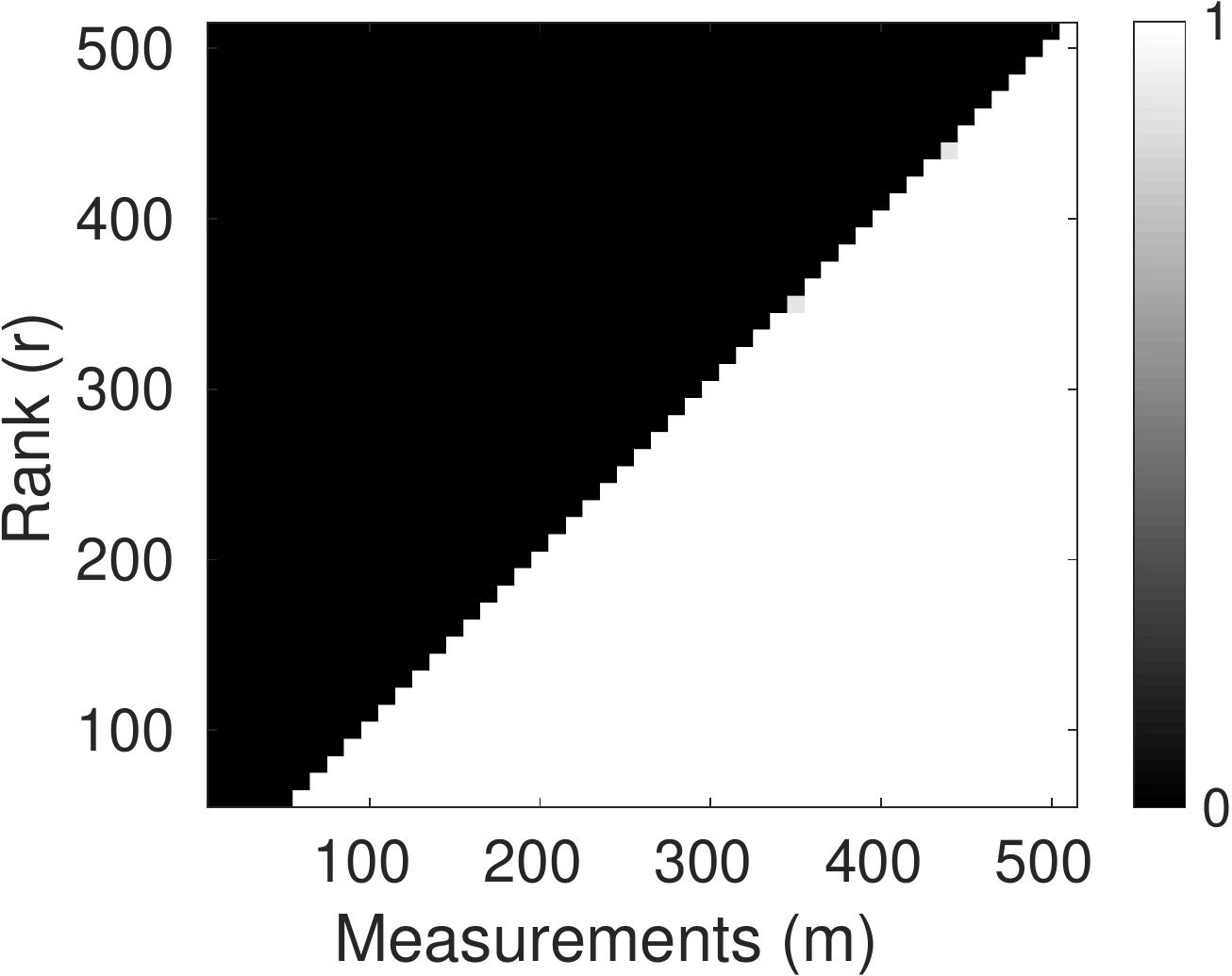}}
\caption{Results on exact recovery for a one-dimensional signal from random Gaussian measurements (see the main text for a description of the experiment). In (a), we show the ground-truth signal $\bxi$ (blue) and the recovered signal $\x^\ast$ (red) for $m = r = 100$, where $r = \rank(\W)$ and $\bxi \in \cR(\W)$. Note that $\bxi$ and $\x^\ast$ coincide exactly. In (b), we plot the empirical probability of exact recovery as a function of $m$ for $4$ different values of $r$. The  empirical probability as a function of both $m$ and $r$ is shown in (c) as a color plot. Note that for every fixed $r$, the empirical probability undergoes a sharp transition from $0$ to $1$ at $m=r$, as predicted by Theorems \ref{thm:improbability} and \ref{thm:exact_rec_prob_gaussian_r<m}.}
\label{fig:exact_rec_gaussian}
\end{figure*}

\section{Proofs of Main Results}
\label{sec:proofs}

In this section, we give the proofs of theorems in Section \ref{sec:main}.
We denote the set of $n \times n$ symmetric, positive semidefinite matrices by $\Sy_+^n$.
Note that from the properties of $\W$ stipulated in Theorem \ref{thm:teodoro}, we get that $\W, (\I - \W) \W$ and $(\I - \W) \W^\dagger$ belong to $\Sy_+^n$; these facts are used in some of the proofs.
A few of the proofs require some results from high-dimensional probability, which are included in Appendix \ref{sec:appendix:auxiliary_results}.

\subsection{Proof of Theorem \ref{thm:improbability}}
\begin{proof}
Consider the following convex program:
\begin{mini}
{}{\Psi(\y) := (\y+\bxi)^\top (\I - \W) \W^\dagger (\y + \bxi)}{}{}
\addConstraint{\y \in \cR(\W) \cap \cN(\A).}
\label{eq:exact-shifted-prob}
\end{mini}
Notice that as $\bxi \in \cR(\W)$, $\bxi$ is a solution of \eqref{eq:noiseless_prob} if and only if $\ze$ is a solution of \eqref{eq:exact-shifted-prob}. 
Using the optimality condition for a convex program, $\ze$ is a solution of \eqref{eq:exact-shifted-prob} if and only if $\nabla \Psi(\ze)^\top (\y - \ze) \geqslant 0$ for all $\y \in \cR(\W) \cap \cN(\A)$; in other words, the equivalent condition is 
\begin{equation*}
\big( (\I - \W) \W^\dagger \bxi \big)^\top \y \geqslant 0 \qquad \forall \y \in \cR(\W) \cap \cN(\A).
\end{equation*}
Note that $\cR(\W) \cap \cN(\A) = \big(\cN(\W) + \cR(\A^\top)\big)^\perp$.
As a result, $\ze$ is a solution of \eqref{eq:exact-shifted-prob} if and only if
\begin{equation}
\label{eq:perp}
(\I - \W) \W^\dagger \bxi \in \Big(\cN(\W) + \cR(\A^\top)\Big).
\end{equation}

Let $\q := (\I - \W) \W^\dagger \bxi$; now, we show that $\q \neq \ze$. 
Suppose $\bxi \in \cN\big((\I-\W)\W^\dagger\big)$.
Note that
\begin{equation*}
\cN\big((\I-\W)\W^\dagger\big) = \cN(\I-\W) \oplus \cN(\W^\dagger);
\end{equation*}
moreover, since $\W \in \Sy^n_+$, we have $\cN(\W^\dagger) = \cN(\W)$.
Thus, there exist unique $\v_1 \in \cN(\I-\W)$ and $\v_2 \in \cN(\W)$ such that $\bxi=\v_1+\v_2$. 
Now, since $\bxi \in \cR(\W)$, note that $\v_2=(\bxi-\v_1) \in \cR(\W) \cap \cN(\W) = \{\ze\}$. 
Therefore, $\bxi=\v_1 \in \cN(\I-\W)$, which is a contradiction.

Let $\rank(\W)=r \leqslant n$ and $\{\u_{r+1},\ldots,\u_n\}$ be a basis of $\cN(\W)$. 
Let $\a_1,\ldots,\a_m$ be the columns of $\A^\top$. Now, we can rewrite \eqref{eq:perp} as follows:
\begin{equation}
\label{eq:linear-dependence}
\q \in \Span(\u_{r+1},\ldots,\u_n) + \Span(\a_1,\ldots,\a_m).
\end{equation}
Notice that to prove the theorem, it suffices to show that \eqref{eq:linear-dependence} holds with probability $0$.
It follows from the eigendecomposition of $\W \in \Sy^n_+$ that $\cR \big((\I - \W)\W^\dagger\big) \subseteq \cR(\W)$.
Subsequently, $\q = (\I-\W)\W^\dagger \bxi \in \cR(\W)=\cN(\W)^\perp$. 
Therefore, $\q, \u_{r+1},\ldots,\u_n$ are linearly independent.
Moreover, since $m + 1 \leqslant r$, the cardinality of the set 
\begin{equation*}
\{ \q, \u_{r+1},\ldots,\u_n, \a_1,\ldots,\a_m \}  
\end{equation*}
is at most $n$. 
Since $\a_1,\ldots,\a_m$ are independent Gaussian random vectors, using Lemma \ref{lem:lin-dependence-general}, the probability that $\q, \u_{r+1},\ldots,\u_n, \a_1,\ldots,\a_m$ are linearly dependent is $0$.
\end{proof}

\subsection{Proof of Theorem \ref{thm:exact_rec_prob_gaussian_r<m}}
\begin{proof}
It follows from \eqref{eq:noiseless_prob} that if $\A|_{\cR(\W)}$ is injective and $\bxi \in \cR(\W)$, then $\bxi$ is the only feasible point of \eqref{eq:noiseless_prob}; thus, it is the unique minimizer \cite[Theorem 3]{Gavaskar2022_robust_recovery}. 
Therefore, it suffices to show that if $m \geqslant \cR(\W)$, then the restriction of a $(m \times n)$ random Gaussian matrix $\A$ to $\cR(\W)$ is injective with probability $1$. 

Let $\rank(\W)=r$ and the columns of $\U \in \Re^{n \times r}$ form of a basis of $\cR(\W)$. 
Let $\a_1^\top,\ldots,\a_m^\top$ be the rows of $\A$; note that $\a_1,\ldots,\a_m$ are independent Gaussian random vectors. 
Now, since $\W$ is independent of $\A$, we have $\U^\top\a_1,\ldots,\U^\top\a_m$ as independent Gaussian random vectors. 
Subsequently, by Lemma \ref{lem:lin-dependence-general}, $\U^\top\a_1,\ldots,\U^\top\a_r$ are linearly independent with probability $1$. 
Therefore, $\rank(\U^\top\A^\top) = \rank(\A\U) = r$ with probability $1$; in other words, the restriction of $\A$ to $\cR(\W)$ is injective with probability $1$.
\end{proof}

\subsection{Proof of Theorem \ref{thm:exact_rec_prob_r<m}}
\begin{proof}
Apply Lemma \ref{lem:subspace_bound} with $\Ur = \cR(\W)$ and $\epsilon = 0.99$.
We get that with probability at least
\begin{equation*}
1 - 2 (12/0.99)^r e^{-m \gamma(0.99/2)} \geqslant 1 - \beta,
\end{equation*}
we have
\begin{equation*}
\frac{\norm{\A \z}}{\norm{\z}} \geqslant 1 - \epsilon > 0 \text{ for all } \z \in \cR(\W) \setminus \{\ZE\},
\end{equation*}
and therefore, $\A|_{\cR(\W)}$ is injective.
Subsequently, the theorem follows from the fact that if the restriction of $\A \in \Re^{m \times n}$ to $\cR(\W)$ is injective and $\bxi \in \cR(\W)$, then $\bxi$ is the unique feasible point and minimizer of \eqref{eq:noiseless_prob}.
\end{proof}

\subsection{Proof of Theorem \ref{thm:robust_rec_prob_r<m}}
\begin{proof}
Let $\hat{\bxi} = \Pi_{\cR(\W)}(\bxi)$, where $\Pi_{\cR(\W)}$ is the orthogonal projection onto $\cR(\W)$.
Note that $(\hat{\bxi} - \bxi)$ is statistically independent of $\A$.
Thus, the inequality in Lemma \ref{lem:gamma} implies that with probability $\geqslant 1 - 2 e^{-m \gamma(\epsilon)}$, we have
\begin{equation}
\label{eq:event_G}
\|\A (\hat{\bxi} - \bxi)\| \leqslant (1 + \epsilon) \|\hat{\bxi} - \bxi\| \leqslant 2 \|\hat{\bxi} - \bxi\|.
\end{equation}
From Lemma \ref{lem:subspace_bound}, by letting $\Ur = \cR(\W)$, we get that with probability $\geqslant 1 - 2 \left(12/\epsilon\right)^r e^{-m \gamma(\epsilon/2)}$,
\begin{equation}
\label{eq:event_F}
(1 - \epsilon) \norm{\z} \leqslant \norm{\A \z} \qquad \forall \z \in \cR(\W).
\end{equation}
Using the union bound, we get that \eqref{eq:event_G} and \eqref{eq:event_F} simultaneously hold with probability $\geqslant 1 - 2 \left(12/\epsilon\right)^r e^{-m \gamma(\epsilon/2)} - 2 e^{-m \gamma(\epsilon)}$.
Since $r \geqslant 1$ and $\gamma$ is an increasing function, 
\begin{equation*}
e^{-m \gamma(\epsilon)} \leqslant e^{-m \gamma(\epsilon/2)} \leqslant \left(12/\epsilon\right)^r e^{-m \gamma(\epsilon/2)};
\end{equation*}
subsequently,
\begin{equation*}
2 \left(12/\epsilon\right)^r e^{-m \gamma(\epsilon/2)} + 2 e^{-m \gamma(\epsilon)} \leqslant 4 \left(12/\epsilon\right)^r e^{-m \gamma(\epsilon/2)} \leqslant \beta,
\end{equation*}
where the last inequality follows from \eqref{eq:m_bound_robust}.
Therefore, \eqref{eq:event_G} and \eqref{eq:event_F} simultaneously hold with probability $\geqslant 1 - \beta$.

Now, in order to complete the proof, we need to show that if \eqref{eq:event_G} and \eqref{eq:event_F} hold, then \eqref{eq:robust_bound_r<m} holds.
Note that
\begin{equation}
\label{eq:split}
\norm{\x^\ast - \bxi} \leqslant \|\x^\ast - \hat{\bxi}\| + \|\hat{\bxi}-\bxi\|.
\end{equation}
Since $\x^\ast$ is a feasible point of \eqref{eq:noisy_prob}, we have $\|\A\x^\ast-\b\| \leqslant \delta$. 
Furthermore, since $\x^\ast - \hat{\bxi} \in \cR(\W)$, using \eqref{eq:event_F} and the triangle inequality, we get
\begin{align}
\label{eq:intermediate-ineq-1}
(1-\epsilon)\|\x^\ast - \hat{\bxi}\| &\leqslant  \|\A \hat{\bxi} - \A \x^\ast\| \nonumber\\
&\leqslant  \|\A \hat{\bxi} - \b\| +  \norm{\b - \A \x^\ast} \nonumber\\
&\leqslant  \|\A \hat{\bxi} - \A \bxi - \boldeta\| + \delta \nonumber\\
&\leqslant  \|\A (\hat{\bxi} - \bxi)\| + (\delta + \norm{\boldeta}).
\end{align}
Now, it follows from \eqref{eq:event_G} and \eqref{eq:intermediate-ineq-1} that
\begin{equation}
\label{eq:intermediate-ineq-2}
\|\x^\ast - \hat{\bxi}\| \leqslant \tfrac{2}{1 - \epsilon} \|\hat{\bxi} - \bxi\| + \tfrac{1}{1 - \epsilon} (\delta + \norm{\boldeta}).
\end{equation}
Since $\dist \big(\bxi, \cR(\W)\big) = \|\hat{\bxi} - \bxi\|$, combining  \eqref{eq:split} and \eqref{eq:intermediate-ineq-2}, we obtain \eqref{eq:robust_bound_r<m}. 
\end{proof}

\section{Numerical Results}
\label{sec:numerical}

In this section, we perform numerical simulations to validate the theoretical results and study the tightness of the bounds in Section \ref{sec:main}. 
We note that comparing the performance with other reconstruction methods is not our aim here, since PnP methods have already been empirically observed to produce state-of-the-art results in several imaging applications \cite{Ahmad2020_PnP_MRI,Yuan2020_PnP_snapshot_CS,Zhang2021_PnP_deep_prior}. Therefore, we solely focus on PnP regularization using linear denoisers; we do not perform extensive comparisons with competing reconstruction techniques, or even PnP using nonlinear denoisers for that matter.

\begin{figure*}[t!]
\centering
\hspace{4.2cm} \includegraphics[width=0.7\linewidth]{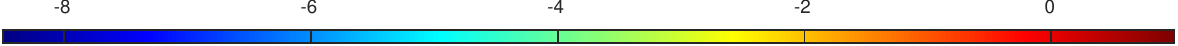}\\
\vspace{-3mm}
\subfloat[$\bxi$.]{\includegraphics[width=0.23\linewidth]{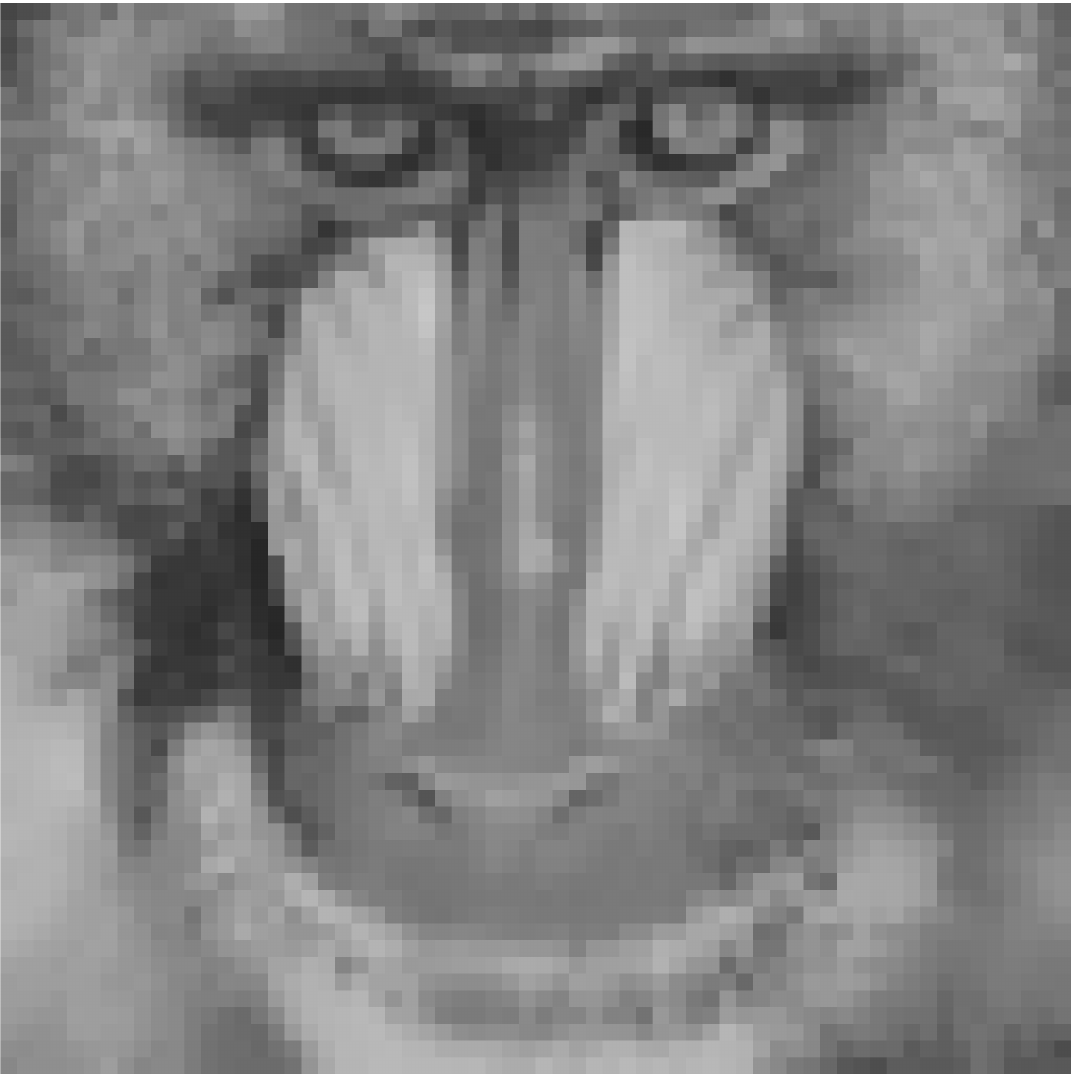}}
\hspace{0.1mm}
\subfloat[$\log_{10} |\x^\ast - \bxi|$, $m=200$.]{\includegraphics[width=0.23\linewidth]{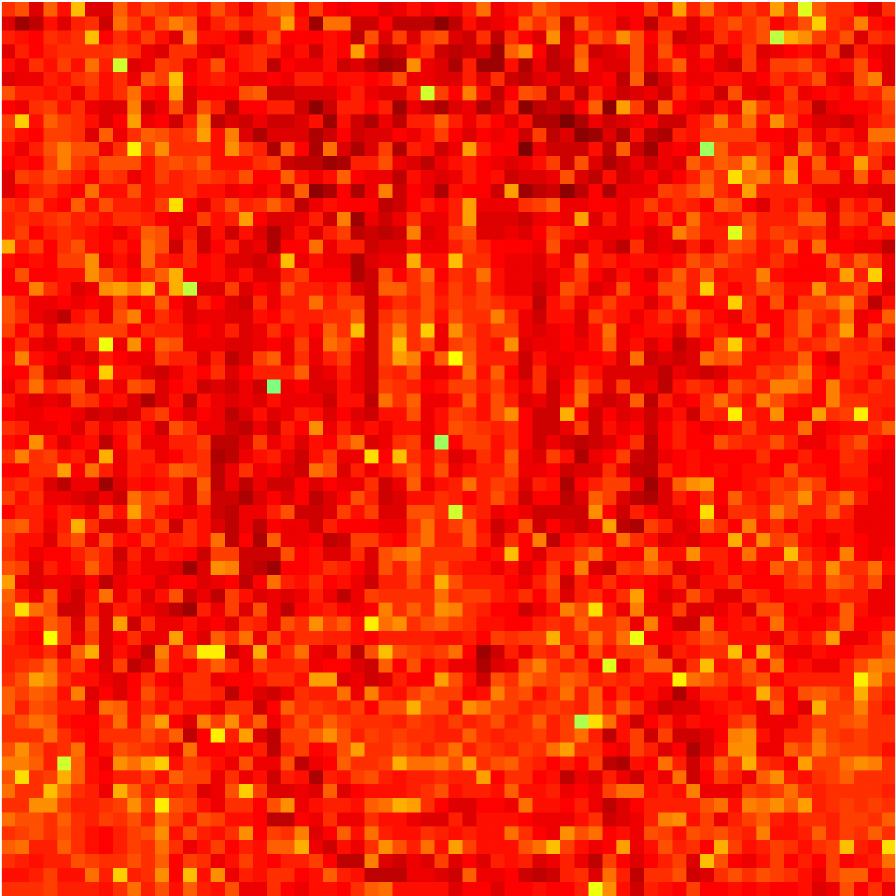}}
\hspace{0.1mm}
\subfloat[$\log_{10} |\x^\ast - \bxi|$, $m=300$.]{\includegraphics[width=0.23\linewidth]{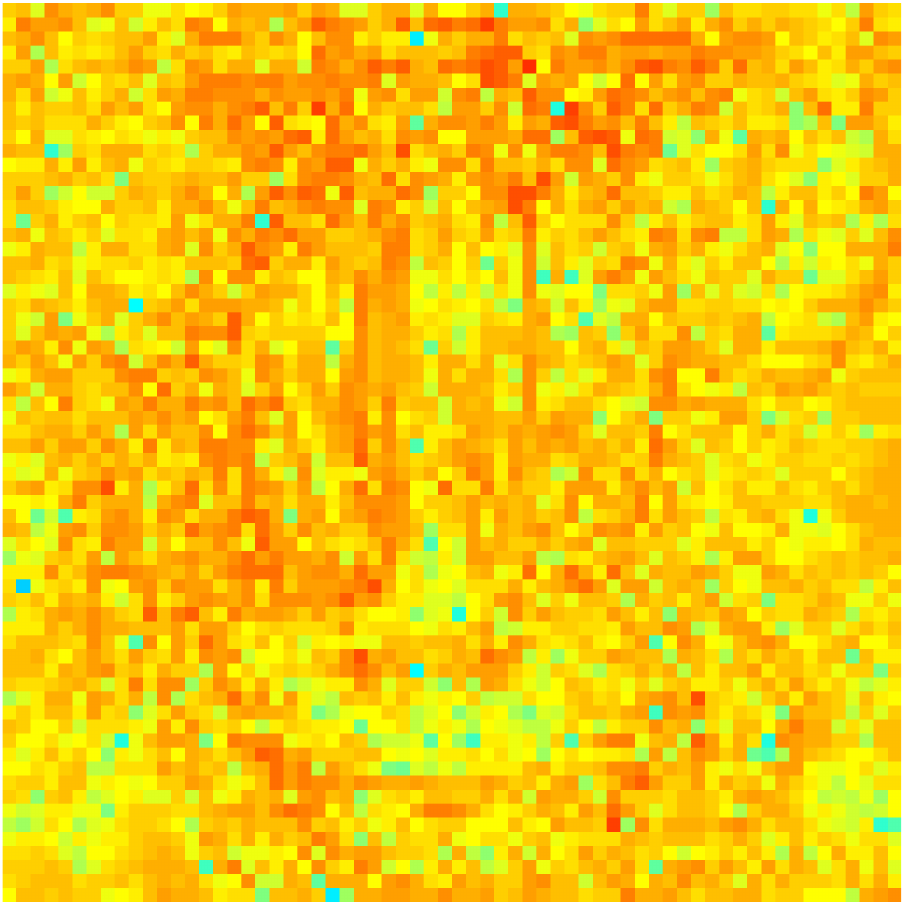}}
\hspace{0.1mm}
\subfloat[$\log_{10} |\x^\ast - \bxi|$, $m=400$.]{\includegraphics[width=0.23\linewidth]{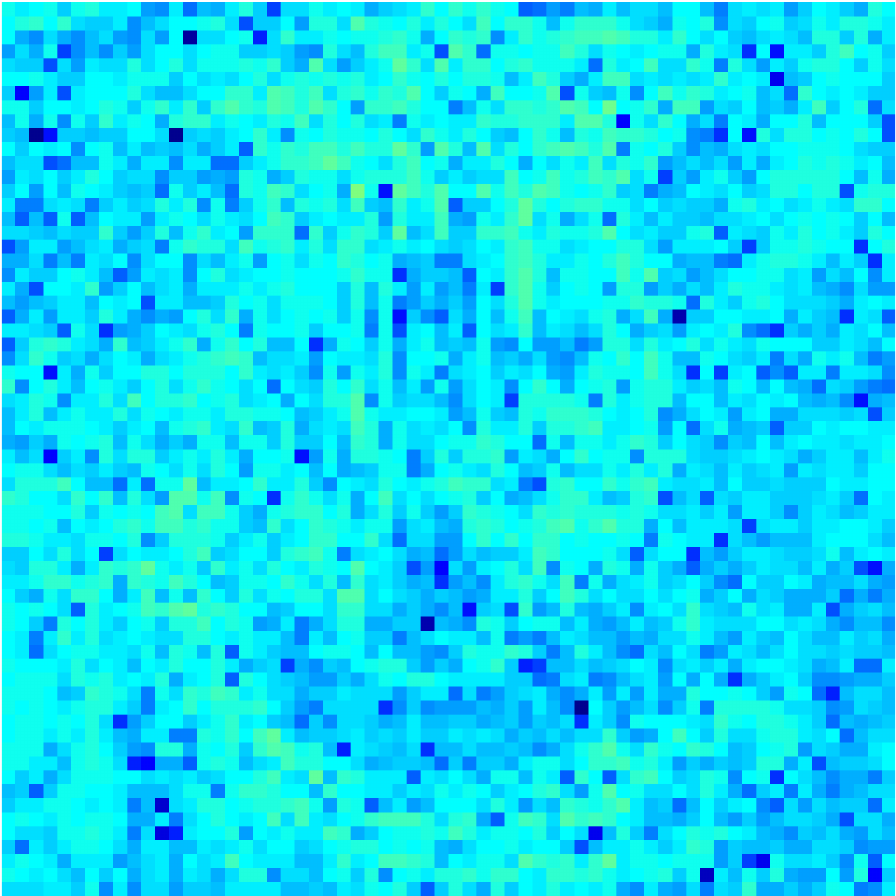}}
\caption{Ground-truth and error images for the recovery of a $64 \times 64$ image from Rademacher measurements with $\rank(\W) = 200$. The images in (b), (c) and (d) are color plots of $\log_{10} |\x^\ast - \bxi|$ for different values of $m$; refer to the colorbar on top. $\A$ is a Rademacher matrix of appropriate size in each case. As expected, the error reduces with increasing $m$.}
\label{fig:rademacher_images}
\end{figure*}

\subsection{Exact Recovery from Gaussian Measurements}

We validate Theorems \ref{thm:improbability} and \ref{thm:exact_rec_prob_gaussian_r<m} in this experiment.
We work with one-dimensional signals throughout, with $n = 512$.
We consider different $\W$'s having ranks ranging from $50$ to $510$ in steps of $10$.
For constructing $\W$ having a specified rank $r$, we take the best rank-$r$ approximation (using SVD) of the DSG-NLM matrix \cite{Sreehari2016_PnP}.
For each $\W$, we ensure that $\bxi \in \cR(\W)$ by applying $\W$ to a scan-line from a natural image; one such $\bxi$ for $r = 100$ is shown in Fig. \ref{fig:exact_rec_gaussian}(a).
$\A$ is taken to be a $(m \times n)$ random Gaussian matrix.
We consider $512$ different values of $m$, ranging from $1$ to $512$.
For each $m$, we generate $100$ random realizations of $\A$.
For each realization, we generate noiseless measurements $\y = \A \bxi$ and record the fraction of times we obtain $\x^\ast = \bxi$ (exact recovery), where $\x^\ast$ is obtained by solving \eqref{eq:noiseless_prob}.
This is the \textit{empirical probability} of exact recovery for the designated values of $r$ and $m$.
This is plotted in Figs. \ref{fig:exact_rec_gaussian}(b) and (c) for different values of $r$ and $m$.
As asserted in Theorems \ref{thm:improbability} and \ref{thm:exact_rec_prob_gaussian_r<m}, we observe exact recovery with probability $0$ when $m < r$ and with probability $1$ when $m \geqslant r$.
This validates the two theorems.

\begin{figure}[t!]
\centering
\includegraphics[width=\linewidth]{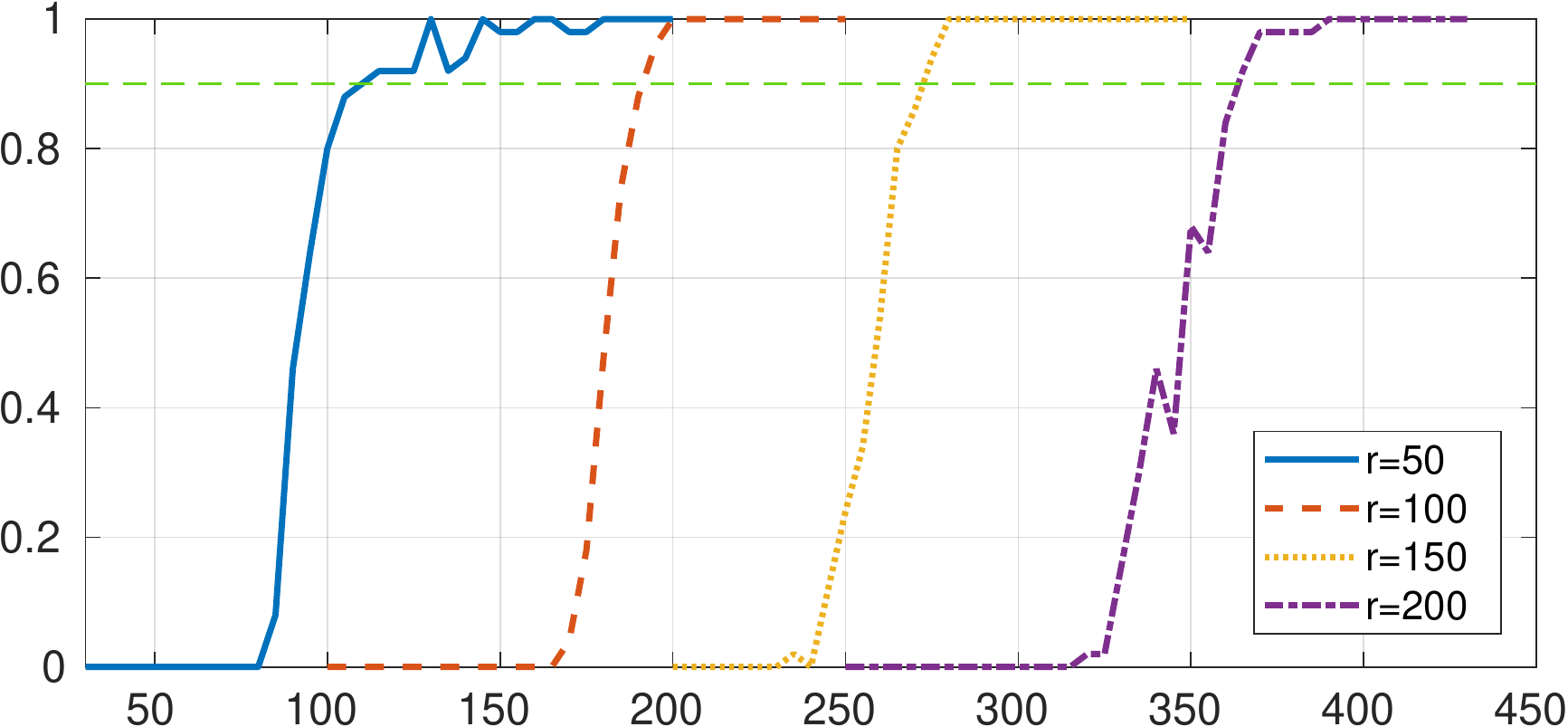}
\caption{Empirical probability of exact recovery (vertical axis) of a $64 \times 64$ image from Rademacher measurements as a function of $m$ (horizontal axis), for different values of $r = \rank(\W)$. The green horizontal line indicates a probability of $0.9$; this is used in Table \ref{tab:exact_rademacher_bound}.}
\label{fig:exact_rec_rademacher}
\end{figure}

\subsection{Exact Recovery from Subgaussian Measurements}
\label{subsec:exact_rademacher}

In this experiment, we compute the empirical probability of exact recovery for the case where $\A$ is a random Rademacher matrix, i.e., each $\A_{ij}$ takes values $\pm 1/\sqrt{n}$ with equal probability.
For this experiment, we fix $\bxi$ to be a $64 \times 64$ image ($n = 4096$) and $\W$ to be the GLIDE filter \cite{Talebi2013_GLIDE}.
Recall that the rank $r$ of GLIDE is user-configurable.
Since the image size (and hence the run-time of the recovery algorithm) is large, we restrict ourselves to $4$ different values of $r$, namely $50, 100, 150, 200$.
For each $r$, we generate $\bxi \in \cR(\W)$ by applying $\W$ to the \textit{Mandril} image.
As an example, the image $\bxi$ for $r = 200$ is shown in Fig. \ref{fig:rademacher_images}(a).
Further, we fix a few different values of $m$ and perform $100$ random trials in which we draw a realization of $\A$, set $\y = \A \bxi$, and obtain $\x^\ast$ by solving \eqref{eq:noiseless_prob}.
The solution is obtained by running $400$ iterations of the CSALSA algorithm \cite{Afonso2010_CSALSA}.
We assume that exact recovery is achieved if the PSNR of the image $\x^\ast$ with respect to $\bxi$ is greater than $80$ dB.
This corresponds to a mean-squared error (MSE) of less than $10^{-8}$ (assuming that the image intensity values are between $0$ and $1$). Thus, for each $r$ and $m$ we record the empirical probability of exact recovery.

A plot of the empirical probability is shown in Fig. \ref{fig:exact_rec_rademacher} as a function of $m$ for different values or $r$.
As expected, for a fixed $r$, the probability increases as $m$ increases.
On the other hand, for a fixed value $p \in (0,1]$, the minimum value of $m$ required to obtain exact recovery with probability at least $p$ increases with $r$.
In Fig. \ref{fig:rademacher_images}, we show an example of the error image $|\x^\ast - \bxi|$ (on a log scale) for $r=200$, for three different values of $m$.
Note that the error decreases as $m$ increases, which is consistent with what we intuitively expect.

To examine the tightness of the lower bound \eqref{eq:m_bound_subgaussian}, we calculate the right side of \eqref{eq:m_bound_subgaussian} for $\beta = 0.1$ and the aforementioned values of $r$; this gives the theoretical minimum number of measurements that guarantee exact recovery with probability at least $0.9$.
The values are noted in Table \ref{tab:exact_rademacher_bound}.
We note that the actual minimum value of $m$ (found using Fig. \ref{fig:exact_rec_rademacher}) is much smaller than the theoretical bound in each case, indicating that the bound in \eqref{eq:m_bound_subgaussian} is quite loose.
Since the image size is small, some of the lower bounds are, in fact, greater than $n$; see the discussion at the end of Section \ref{sec:main}.

\begin{table}[t!]
\centering
\caption{Theoretical and empirical lower bounds on $m$ for achieving exact recovery with probability $\geqslant 0.9$ from Rademacher measurements, for different values of $r = \rank(\W)$. The theoretical bound is the right side of \eqref{eq:m_bound_subgaussian}, and the empirical bound is approximately found from the plots in Fig. \ref{fig:exact_rec_rademacher}.}
\label{tab:exact_rademacher_bound}
\begin{tabular}{lcccc}
\toprule
$r$ & $50$ & $100$ & $150$ & $200$ \\
\midrule
$m$ (Theoretical) & $3113$ & $6152$ & $9192$ & $12231$ \\
$m$ (Empirical) & $120$ & $190$ & $280$ & $370$ \\
\bottomrule
\end{tabular}
\end{table}

\begin{figure*}[t]
\centering
\includegraphics[width=\linewidth]{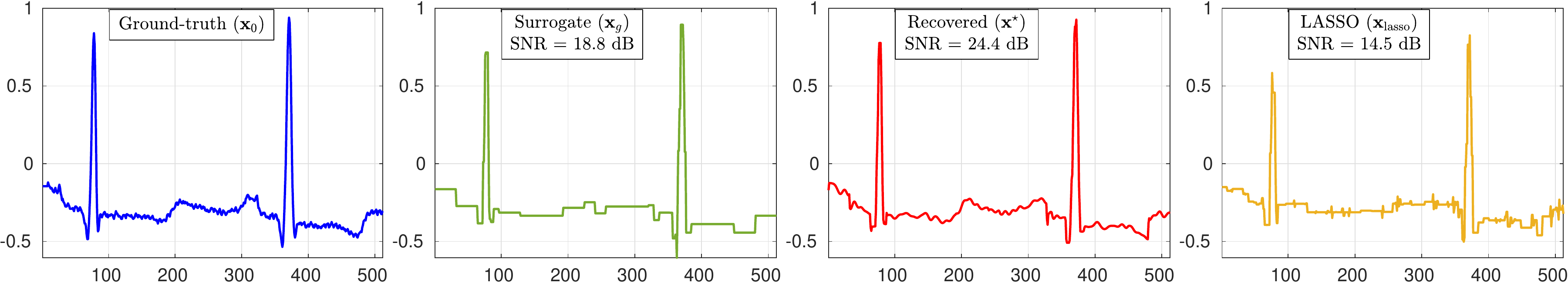}\\
\caption{ECG signal recovery from $m=150$ random Gaussian noisy measurements. The signal length is $n=512$. The surrogate signal $\x_g$ is obtained using CoSaMP \cite{Needell2009_cosamp}, and is used to construct $\W$. The final estimate $\x^\ast$ is then found by solving \eqref{eq:noisy_prob}. The reconstruction using LASSO is shown for comparison.}
\label{fig:ecg}
\end{figure*}

\subsection{Robust Recovery}
\label{subsec:robust}

Theorem \ref{thm:robust_rec_prob_r<m} implies that if we fix the parameters $\beta$, $\epsilon$, $\boldeta$ and $\delta$, then the bound \eqref{eq:robust_bound_r<m} holds with probability at least $1 - \beta$ across different realizations of $\A$, provided $m$ is sufficiently large.
We take $\A$ to be a Rademacher matrix and $\bxi$ as the \textit{Mandril} image (resized to $64 \times 64$).
$\W$ is the GLIDE filter computed using $\bxi$ as the guide image.
Since the guide image is $\bxi$ itself, we expect that $\bxi \notin \cR(\W)$; indeed, we verified this numerically by computing the distance of $\bxi$ from $\cR(\W)$.
We fix $\beta = 0.1$, $\boldeta$ as Gaussian noise with variance $0.05^2$, $\delta = 1.2 \norm{\boldeta}$ (to ensure that problem \eqref{eq:noisy_prob} is feasible), and $\epsilon = 0.8$.
Plugging these values into Theorem \ref{thm:robust_rec_prob_r<m}, we get that with probability at least $0.9$, the error $\norm{\x^\ast - \bxi}$ is less than the right side of \eqref{eq:robust_bound_r<m} if
\begin{equation*}
m \geqslant 125.75 + 92.32 r.
\end{equation*}
For $r = 50, 100, 150, 200$, the right side of the above inequality evaluates to $4742, 9358, 13924$ and $18590$.
In practice, by conducting $100$ random trials, we observed that for all four values of $r$, \eqref{eq:robust_bound_r<m} holds with probability $1$ even for $m$ as low as $500$ (and higher).
For the case $r=200$, the right side of \eqref{eq:robust_bound_r<m} evaluates to $\approx 3000$, whereas the average value of the left side over the $100$ trials is $\approx 370$.
Thus, the bound in \eqref{eq:m_bound_robust}, as well as the probability bound $1 - \beta$, are observed to be loose in practice.

\subsection{Application: ECG Signal Recovery}

While PnP has mostly been used for imaging applications in the past, it can in principle be used for compressed sensing of other signals. In particular, we show how it can be used for ECG signal recovery from compressed Gaussian measurements, where the reconstruction is performed by solving \eqref{eq:noisy_prob}.
We take the ground-truth signal $\bxi$ as the first $512$ samples of an ECG signal from the MIT-BIH Arrhythmia Database \cite{Moody2001_MIT_ECG_database}, i.e., $n = 512$.
For $m = 150$, we generate the measurement vector $\b = \A \bxi + \boldeta$, where $\A \in \Re^{m \times n}$ is a random Gaussian matrix and $\boldeta$ is Gaussian noise having standard deviation $5 \times 10^{-3}$.
$\W$ is taken to be the SVD-based low-rank approximation of DSG-NLM with rank $150$.
The DSG-NLM denoiser is computed from a guide (surrogate) signal (see Section \ref{subsec:independence}); we obtain a surrogate signal by running $20$ iterations of the computationally efficient CoSaMP algorithm \cite{Needell2009_cosamp}.
The final solution $\x^\ast$ is then obtained by solving \eqref{eq:noisy_prob}, where feasibility is ensured by setting $\delta = 2 \norm{\boldeta}$.
The result is shown in Fig. \ref{fig:ecg}, along with the signal-to-noise ratio (SNR) of $\x^\ast$ with respect to $\bxi$; for comparison, we also show the reconstruction obtained using the LASSO algorithm ($\ell_1$ minimization) \cite{Candes2008_compressive_sampling}. 
Although we do not advocate the superiority of PnP over existing methods, it is evident from the SNR levels that its reconstruction quality is better than LASSO.

\section{Discussion}
\label{sec:discussion}

\subsection{Relation to Existing Work}
Our results are similar in spirit to those in classical compressed sensing \cite{Candes2006_stable_recovery,Candes2008_compressive_sampling,Shwartz2014_UML}.
For example, one of the central results in compressed sensing is as follows: An $r$-sparse signal in $\Re^n$ (i.e., a signal having at most $r$ non-zero samples) can, with high probability, be recovered exactly from $m$ noiseless random Gaussian measurements if $m \geqslant O(r \log n)$ \cite[Sec. 23.3]{Shwartz2014_UML}.
On the other hand, our result requires $m \geqslant \rank(\W)$ for exact recovery from Gaussian measurements.
In this sense, $\rank(\W)$ plays a similar role to sparsity in classical compressed sensing.

For PnP regularization, however, a probabilistic analysis of exact and robust recovery has not been attempted before to the best of our knowledge.
The papers \cite{Liu2021_PnP_REC} and \cite{Bora2017_CSGM} are somewhat related to the current work.
In \cite{Liu2021_PnP_REC}, error bounds are established for images recovered from compressive measurements using the PnP-ISTA algorithm.
The main difference compared to our work is that \cite{Liu2021_PnP_REC} takes a purely algorithmic approach, whereas we view the recovery problem from an optimization perspective using the explicit PnP regularizer $\Phi_{\W}$.
Moreover, probabilistic guarantees are not given in \cite{Liu2021_PnP_REC}.
In \cite{Bora2017_CSGM}, the recovered image is taken to be the minimizer of $\norm{\A\x - \b}^2$, where the feasible set is the range of a generative model such as a generative adversarial network or variational auto-encoder.
Our work is similar in that we also require the reconstruction to lie in the range of a denoiser.
However, there is no obvious direct relationship between our work and \cite{Bora2017_CSGM}.
Another difference is the denoisers considered---while we work with linear denoisers, \cite{Liu2021_PnP_REC,Bora2017_CSGM} use neural networks as the denoiser or generative model.

\subsection{Extension to Randomized Fourier Measurements}
\label{subsec:Fourier-Hadamard}

Theorems \ref{thm:exact_rec_prob_r<m} and \ref{thm:robust_rec_prob_r<m} apply to sensing matrices that satisfy the concentration inequality in Lemma \ref{lem:gamma}.
The well-known Johnson-Lindenstrauss (JL) Lemma, stated below as Lemma \ref{lem:JL}, is a generalization of Lemma \ref{lem:gamma} to a finite set of points as opposed to a single point \cite{Baraniuk2008_RIP_proof,Matouvsek2008_JL_variants,Shwartz2014_UML}.
Thus, we can conclude that the recovery guarantees in Theorems \ref{thm:exact_rec_prob_r<m} and \ref{thm:robust_rec_prob_r<m} apply in general to sensing matrices satisfying the JL Lemma.
\begin{lemma}
\label{lem:JL}
Let $\A$ be a random subgaussian matrix, and $Q \subseteq \Re^n$ be a finite set of unit vectors (w.r.t. the $\ell_2$ norm) that are independent of $\A$.
Then for any $\epsilon \in (0,1)$, with probability at least $1 - 2 |Q| e^{- m \gamma(\epsilon)}$ we have
\begin{equation*}
1 - \epsilon \leqslant \norm{\A \x}^2 \leqslant 1 + \epsilon \text{ for all } \x \in Q,
\end{equation*}
where $\gamma$ is the function in Lemma \ref{lem:gamma}.
\end{lemma}

For signals having a large number of samples (e.g. images), random Fourier or Hadamard measurements are computationally more efficient than (sub)-Gaussian measurements \cite{Ailon2013_JL_hadamard}. The sensing matrix for the former can be written as
\begin{equation}
\label{eq:A_hadamard}
\A = \frac{1}{\sqrt{m}} \S \F \D,
\end{equation}
where $\S \in \Re^{m \times n}$ is a random subset of $m$ rows of the $n \times n$ identity matrix (a random sampling operator), $\F \in \Co^{n \times n}$ is an orthogonal transform such as the unnormalized discrete Fourier or Walsh-Hadamard transform, and $\D \in \Re^{n \times n}$ is a random diagonal matrix with diagonal entries drawn uniformly from $\{-1,1\}$.
Unlike  Gaussian measurements, Fourier and Hadamard transforms can be computed in $O(n \log n)$ time without storing the matrix $\F$ \cite{Ailon2013_JL_hadamard}.
This is particularly useful for images.
The following is a restatement of \cite[Theorems 2.1 and 3.1]{Ailon2013_JL_hadamard}, which asserts that \eqref{eq:A_hadamard} satisfies a somewhat different version of the JL Lemma.
\begin{lemma}
\label{lem:JL_hadamard}
Let $\A$ be a random $(m \times n)$ matrix in \eqref{eq:A_hadamard}.
Let $Q$ be a finite subset of the unit ball in $\Re^n$, that is independent of $\A$.
If $m = O \big(\epsilon^{-4} (\log |Q|) (\log^4 n) \big)$, then for any $\epsilon \in (0,1)$, with probability at least $0.98 \times 0.99$ we have
\begin{equation*}
1 - O(\epsilon) \leqslant \norm{\A \x} \leqslant 1 + O(\epsilon)
\end{equation*}
uniformly for all $\x \in Q$.
\end{lemma}

\begin{figure}[t]
\centering
\includegraphics[width=\linewidth]{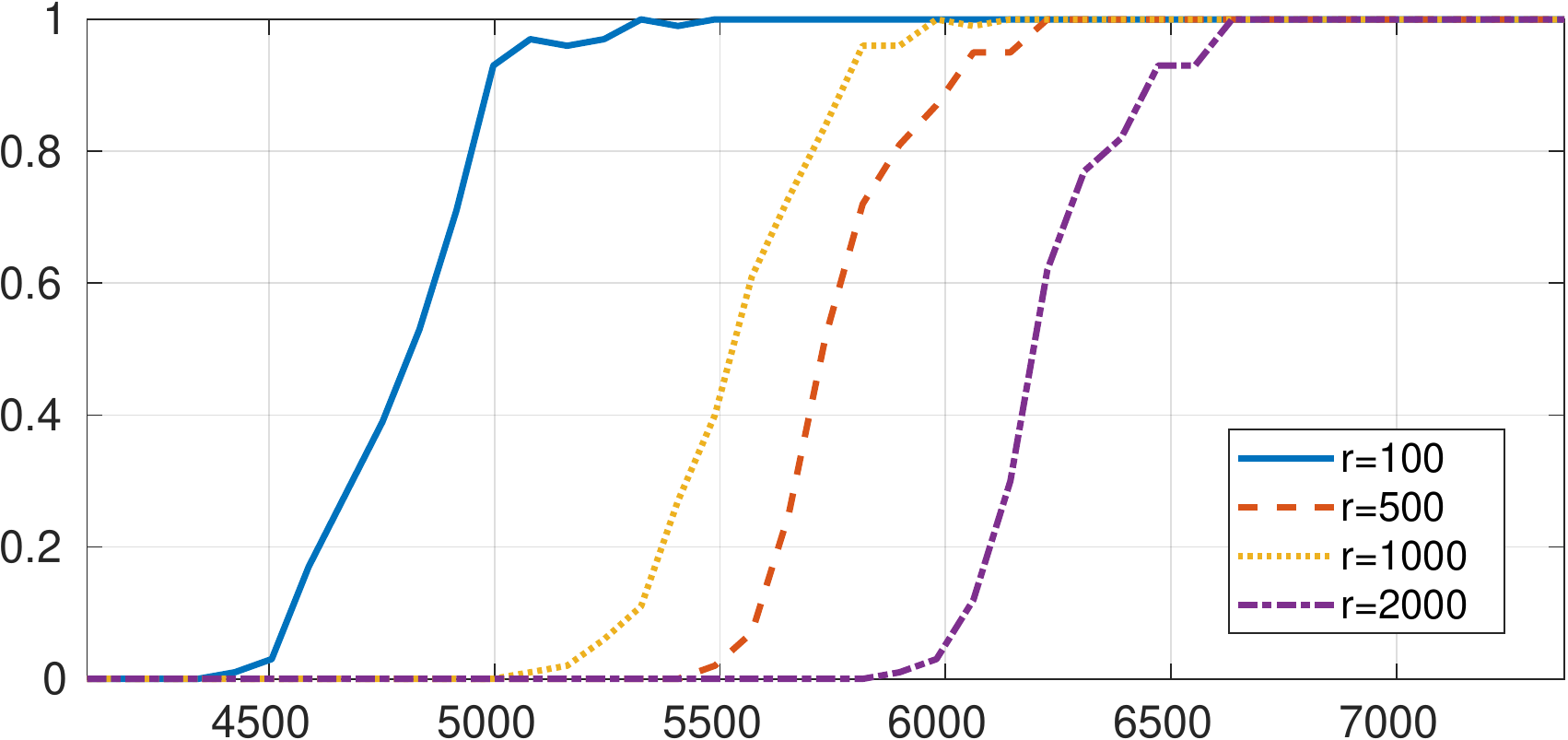}
\caption{Empirical probability of exact recovery (vertical axis) of a $128 \times 128$ image from randomized Fourier measurements as a function of $m$ (horizontal axis) for different values of $r = \rank(\W)$.}
\label{fig:exact_rec_fourier}
\end{figure}

\begin{figure}[t]
\centering
\subfloat[$\bxi$.]{\includegraphics[width=0.33\linewidth]{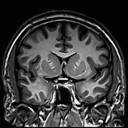}}
\hfill
\subfloat[$\x^\ast$, $m = 200$.]{\includegraphics[width=0.33\linewidth]{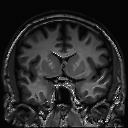}}
\hfill
\subfloat[$\x^\ast$, $m = 6000$.]{\includegraphics[width=0.33\linewidth]{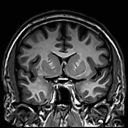}}
\caption{Ground-truth and examples of reconstructed images for the experiment in Fig. \ref{fig:exact_rec_fourier} with $\rank(\W) = 1000$, where $\bxi \in \cR(\W)$. $\A$ is a random realization of \eqref{eq:A_hadamard} in each of the examples in (b) and (c). The image size is $128 \times 128$.}
\label{fig:fourier_images}
\end{figure}

Thus, in principle, it could be possible to derive probabilistic guarantees for exact and robust recovery for the sensing model in \eqref{eq:A_hadamard}.
However, since the  hidden constants in the above $O(\cdot)$ notation are not explicitly given in \cite{Ailon2013_JL_hadamard}, it is difficult to derive an analogue of Theorems \ref{thm:exact_rec_prob_r<m} and \ref{thm:robust_rec_prob_r<m} for this model.
Nevertheless, we can numerically verify that exact recovery can be achieved if $m$ is large enough.
We perform a similar experiment to that in Section \ref{subsec:exact_rademacher}, but for randomized Fourier measurements.
We take the image size to be $128 \times 128$ ($n = 16384$).
A plot of the empirical probabilities of exact recovery is shown in Fig. \ref{fig:exact_rec_fourier} for four different values of $r$.
Note that the general trend is similar to what we expect, i.e., more measurements are required for a higher probability of exact recovery.
We leave a rigorous analysis of this observation for future work.
Fig. \ref{fig:fourier_images} shows a visual example of the recovered images.

\subsection{Independence of $\W$ and $\A$}
\label{subsec:independence}

Note that in the theorems in Section \ref{sec:main}, we require $\W$ to be statistically independent of the random matrix $\A$.
Practical denoisers such as DSG-NLM and GLIDE require access to a guide signal/image to populate $\W$.
In the experiments in Section \ref{sec:numerical}, we constructed $\W$ from some fixed guide signal, and generated random observations via $\A$ independently of $\W$.
This automatically ensured the independence of $\W$ and $\A$.
We did this because we were interested in observing the recovery behavior when $\W$ is fixed and $\A$ is random, and therefore, it was necessary to fix a common $\W$ for all random realizations of $\A$.
However, in practical CS reconstruction scenarios, we do not a priori have access to a guide signal.
Instead, the guide signal is obtained by applying preprocessing techniques to the observation $\b$.
For example, we can run a small number of PnP iterations, say $l$, in which $\W$ is generated using the image in the previous iteration, and then keep $\W$ fixed from the $(l+1)^{\rm th}$ iteration onward, e.g., see \cite{Sreehari2016_PnP,Gavaskar2020_PnP-ISTA_kernel,Gavaskar2021_PnP_linear_denoisers,Nair2021_PnP_fixed_point}.
The guide signal is thus the image in the $l^{\rm th}$ iteration, which indirectly depends on $\b$, and hence on $\A$ (since $\b = \A \bxi + \boldeta$).
Therefore, strictly speaking, $\W$ is not independent of $\A$.
However, the relationship between $\W$ and $\A$ is complicated due to the technique used to generate the guide signal.
The upside is that the statistical independence of $\W$ and $\A$ seems to be a reasonable assumption in practice.
This is similar to the following claim in \cite{Milanfar2013_filtering_tour}: for image denoising, computing $\W$ from a pre-filtered version of the noisy image (as opposed to the noisy image itself) largely removes the statistical dependence of $\W$ on the noise.

\subsection{Role of $\cR(\W)$}
\label{subsec:subspace_prior}

Note that the lower bounds on $m$ in Theorems \ref{thm:exact_rec_prob_gaussian_r<m}, \ref{thm:exact_rec_prob_r<m} and \ref{thm:robust_rec_prob_r<m} depend only on $r$, the dimension of $\Img(\W)$, and not on the ambient dimension $n$.
In contrast, in classical compressed sensing, for exactly recovering an $r$-sparse signal we need $m \geqslant O(r \log n)$ \cite[Sec. 23.3]{Shwartz2014_UML}; note that the lower bound on $m$ depends on $n$.
The reason that $n$ appears in this bound can be attributed to the fact that the set of $r$-sparse signals is the union of $\binom{n}{r}$ subspaces of dimension $r$.
On the other hand, the PnP regularizer \eqref{eq:kernel-regularizer} forces the solution of \eqref{eq:noisy_prob} and \eqref{eq:noiseless_prob} to lie in $\Img(\W)$.
Thus, in our analysis, we need to consider only signals in the subspace $\Img(\W)$ instead of a union of subspaces.
Since the dimension of $\Img(\W)$ has no relation with the ambient dimension $n$, the lower bounds on $m$ in Theorems \ref{thm:exact_rec_prob_gaussian_r<m}, \ref{thm:exact_rec_prob_r<m} and \ref{thm:robust_rec_prob_r<m} depend only on $r$ and not on $n$.

On a related note, the prior that the ground-truth is in $\Img(\W)$ is strong enough to yield non-trivial recovery guarantees in the form of Theorems \ref{thm:exact_rec_prob_gaussian_r<m}, \ref{thm:exact_rec_prob_r<m} and \ref{thm:robust_rec_prob_r<m} without explicitly using the real-valued part of the objective function, $\x^\top (\I - \W) \W^\dagger \x$.
Indeed, it can be observed from the proofs that the only information about $\Phi_{\W}$ that is used is that it is infinite outside $\Img(\W)$, implying that the solution lies in $\Img(\W)$.
In this aspect, our recovery guarantees are similar to \cite{Bora2017_CSGM}.
In \cite{Bora2017_CSGM}, the reconstructed signal is regularized by requiring it to lie in the range of a generative model instead of using an explicit regularization function.
However, we note that the real-valued component of the objective function $\Phi_{\W}$ is used in the proof of Theorem \ref{thm:improbability}.

Recall from the experiments in Sections \ref{subsec:exact_rademacher} and \ref{subsec:robust} that the bounds on $m$ in Theorems \ref{thm:exact_rec_prob_r<m} and \ref{thm:robust_rec_prob_r<m} are loose.
A possible avenue to make these tighter could be using the properties of the real-valued component $\x^\top (\I - \W) \W^\dagger \x$.
However, this is a non-trivial task and is left for future work.

\subsection{Closing Remarks and Future Work}
\label{sec:closing_remarks}

We can extend the results in Section \ref{sec:main} to {\em proximable} non-symmetric linear denoisers characterized in \cite{Gavaskar2021_PnP_linear_denoisers}.
For non-symmetric kernel filters of the form $\W = \D^{-1} \K$ (e.g., NLM, bilateral filter and LARK \cite{Milanfar2013_filtering_tour}), where the normalization matrix $\D \in \Sy^n_{++}$ and kernel matrix $\K \in \Sy^n_+$, the spectrum of $\W$ lies in $[0,1]$ and $\W$ is semisimple \cite{Gavaskar2021_PnP_linear_denoisers}. 
Let $\V \LA \V^{-1}$ be an eigenvalue decomposition of $\W$ such that $\V\V^\top=\D^{-1}$ \cite[Sec. IV-B]{Gavaskar2021_PnP_linear_denoisers}. 
Subsequently, using the exposition in \cite{Gavaskar2021_PnP_linear_denoisers} and \cite{Nair2021_PnP_fixed_point}, we can associate the following regularizer (up to a scalar multiplication) with $\W$:
\begin{equation*}
\Phi_{\W}(\x) =
\begin{cases}
\tfrac{1}{2} \big\langle  (\I - \W)\x, \W^{\mathfrak{g}} \x \big\rangle_\D, & \mathrm{if} \ \x \in \cR(\W),\\
\infty, & \mathrm{otherwise},
\end{cases}
\end{equation*}
where $\W^{\mathfrak{g}} := \V \LA^\dagger \V^{-1}$ is a {\em reflexive generalized inverse} of $\W$ \cite[Def. 2]{Deutsch1971_semi_inverses} and $\langle \cdot\,,\cdot \rangle_\D$ is the inner-product w.r.t. $\D \in \Sy^n_{++}$. 

At the end of Sections \ref{subsec:Fourier-Hadamard} and \ref{subsec:subspace_prior}, we have discussed a few open questions arising from the current work and possible directions of future research. Furthermore, an interesting question that will be explored in an upcoming work is bounding the recovery error for general linear inverse problems such as deblurring, inpainting, and superresolution.
In this regard, it can be shown that under suitable conditions,
\begin{equation}
\label{eq:robust_bound_high_rank}
\|\x^\ast - \bxi\| \leqslant c(\W,\A) \sqrt{d(\bxi,\W,\A)^2 + \big( \| \boldeta \| + \delta \big)^2},
\end{equation}
where $\x^\ast$ is the minimizer of \eqref{eq:noisy_prob}.
Interestingly, this bound requires that $\rank(\W) \geqslant m$.
The difficulty with this bound is that for the compressed sensing problem, $c$ and $d$ in \eqref{eq:robust_bound_high_rank} are random variables since they depend on $\A$.
Subsequently, unlike \eqref{eq:robust_bound_r<m}, we do not get a global bound on the recovery error from \eqref{eq:robust_bound_high_rank}. 
This question will be investigated in future work.

\appendix
\section{}
\subsection{Auxiliary Results}
\label{sec:appendix:auxiliary_results}

We state a couple of auxiliary results which are used in Section \ref{sec:proofs}.
\begin{lemma}
\label{lem:lin-dependence-general}
Let $\y_1,\ldots,\y_k$ be independent $\Re^n$-valued random vectors, where $k \leqslant n$, with distributions which are absolutely continuous with respect to the Lebesgue measure on $\Re^n$.
Let $\v_1,\ldots,\v_l \in \Re^n$ be linearly independent vectors, where $k+l \leqslant n$. 
Let $B$ be the event that $\v_1,\ldots,\v_l,\y_1,\ldots,\y_k$ are linearly dependent. 
Then $\Prob(B)=0$.
\end{lemma}
\begin{proof}
Since $\v_1,\ldots,\v_l$ are linearly independent, the event $B$ occurs if and only if there exists $i \in \{1,\ldots,k\}$ such that $\y_i \in \Span(\v_1,\ldots,\v_l,\y_1,\ldots,\y_{i-1})$.
Let $p=\Prob(B)$; note that
\begin{align*}
p&=\Prob \bigg(\bigcup_{i=1}^k \big[\y_i \in \Span(\v_1,\ldots,\v_l,\y_1,\ldots,\y_{i-1})\big] \bigg) \\
&\leqslant \sum_{i=1}^k \Prob\big[\y_i \in \Span(\v_1,\ldots,\v_l,\y_1,\ldots,\y_{i-1})\big].
\end{align*}
Let $p_i$ denote the $i^{\rm th}$ term in the above sum.
To show $p=0$, it suffices to show that each $p_i=0$.
Notice that for $i \geqslant 2$,
\begin{align}
\label{eq:exp-of-prob-general}
p_i = \E \Big[\!\Prob\!\big[\y_i \in \Span(\v_1,\ldots,\v_l,\y_1&,\ldots,\y_{i-1}) \big| \nonumber \\ &\y_1,\ldots,\y_{i-1}\big]\!\Big], 
\end{align}
where the expectation is with respect to $\y_1,\ldots,\y_{i-1}$.
Note that for arbitrary $\c_1, \ldots, \c_{i-1} \in \Re^n$, since $\y_1,\ldots,\y_k$ are independent random vectors,
\begin{align}
&\Prob \Big[\y_i \in \Span(\v_1,\ldots,\v_l,\y_1,\ldots,\y_{i-1}) \big| \nonumber \\
& \hspace{4cm} \y_1=\c_1,\ldots,\y_{i-1}=\c_{i-1} \Big] \nonumber\\
&= \Prob\big[\y_i \in \Span(\v_1,\ldots,\v_l,\c_1,\ldots,\c_{i-1})\big]. \label{eq:probability-span-general}
\end{align}
Moreover, since $1 \leqslant i \leqslant k$ and $k+l \leqslant n$, the dimension of $\Span(\v_1,\ldots,\v_l,\c_1,\ldots,\c_{i-1}) \subseteq \Re^n$ is strictly less than $n$, and hence its Lebesgue measure is zero. 
Now, since $\y_1,\ldots,\y_k$ are absolutely continuous random vectors, it follows from \eqref{eq:exp-of-prob-general} and \eqref{eq:probability-span-general} that $p_i = 0$ for $i \geqslant 2$.
A similar argument holds for the case $i=1$.
\end{proof}
The following Lemma is a straightforward generalization of \cite[Lemma 5.1]{Baraniuk2008_RIP_proof}, and it can be proved along the same lines. 
The only difference is that in Lemma \ref{lem:subspace_bound}, $\Ur$ is an arbitrary subspace of $\Re^n$, whereas \cite[Lemma 5.1]{Baraniuk2008_RIP_proof} focuses on the case where $\Ur$ is a canonical subspace.
\begin{lemma}
\label{lem:subspace_bound}
Let $\A$ be a random subgaussian matrix.
Let $\Ur$ be a fixed subspace of $\Re^n$ that is independent of $\A$, such that $\dim \Ur = r < n$.
Then, for any $\epsilon \in (0,1)$, with probability at least $1 - 2 (12/\epsilon)^r e^{-m \gamma(\epsilon/2)}$ we have
\begin{equation}
\label{eq:subspace_bound}
(1 - \epsilon) \norm{\x} \leqslant \norm{\A \x} \leqslant (1 + \epsilon) \norm{\x} \qquad \forall \x \in \Ur,
\end{equation}
where $\gamma$ is the function in Lemma \ref{lem:gamma}.
\end{lemma}

\subsection{When Does \eqref{eq:m_bound_robust} Hold?}
\label{sec:minimum_n}

Let $L\colon (0,1] \times (0,1] \to \Re_+$ be defined as follows:
\begin{equation}
\label{eq:def_L}
L(\beta, \epsilon) := \frac{\ln(4/\beta) + r \ln(12/\epsilon)}{\gamma(\epsilon/2)};
\end{equation} 
this gives the lower bound in \eqref{eq:m_bound_robust}.
Recall from Section \ref{sec:main} that for random subgaussian matrices, $\gamma$ is a continuous and strictly increasing function.
Notice that the numerator and denominator in \eqref{eq:def_L} are strictly decreasing and increasing functions of $\epsilon$. 
Furthermore, 
\begin{equation*}
\frac{\partial L}{\partial \beta} = \frac{-1/\beta}{\gamma(\epsilon/2)} < 0.
\end{equation*}
Therefore, $L$ is a strictly decreasing function in each variable. 
Consequently,
\begin{equation}
\label{eq:inf_L}
\inf_{\beta,\epsilon \in (0,1]} L(\beta, \epsilon) = L(1,1).
\end{equation}
Moreover, note that for arbitrary $\hat{\beta},\hat{\epsilon} \in (0,1]$, we have 
\begin{equation}
\label{eq:zero-coercivity of L}
\lim_{\beta \to 0} L(\beta, \hat{\epsilon}) = \lim_{\epsilon \to 0} L(\hat{\beta}, \epsilon) = \infty.
\end{equation}

The following lemma asserts that if $n$ is sufficiently large, then we can get legitimate lower bounds on $m$ from \eqref{eq:m_bound_robust}.
\begin{proposition}
\label{prp:lower-bound on m}
If $n > L(1,1)$, then there exist unique $\beta_0, \epsilon_0 \in (0,1)$ such that $L(\beta_0,1)=L(1,\epsilon_0)=n$. Furthermore,
\begin{enumerate}[label=(\roman*)]
\item Given $\beta_1 \in (\beta_0, 1)$, there exists unique $\epsilon_1 \in (\epsilon_0, 1)$ such that $L(\beta_1,\epsilon_1) = n$ and $L(\beta_1,\epsilon) < n$ for all $\epsilon \in (\epsilon_1, 1)$. 
\item Given $\epsilon_1 \in (\epsilon_0, 1)$, there exists unique $\beta_1 \in (\beta_0, 1)$ such that $L(\beta_1,\epsilon_1) = n$ and $L(\beta,\epsilon_1) < n$ for all $\beta \in (\beta_1, 1)$.
\end{enumerate}
\end{proposition}
\begin{proof}
Since $n > L(1,1)$, using \eqref{eq:inf_L}, \eqref{eq:zero-coercivity of L} along with the fact that $L$ is continuous and a strictly decreasing function in each variable, we can conclude that there exist unique $\beta_0, \epsilon_0 \in (0,1)$ such that $L(\beta_0,1)=L(1,\epsilon_0)=n$.

Next, we prove the statement-(i); notice that statement-(ii) can be proved along similar lines. 
Given $\beta_1 \in (\beta_0, 1)$, since $L$ is a strictly decreasing function in each variable, we have 
\begin{align*}
L(\beta_1,1) &< L(\beta_0,1) = n, \\
L(\beta_1,\epsilon_0) &> L(1,\epsilon_0) = n.
\end{align*}
Subsequently, using the fact $L$ is continuous and a strictly decreasing function in the second variable, we can conclude that there exists unique $\epsilon_1 \in (\epsilon_0, 1)$ such that $L(\beta_1,\epsilon_1) = n$ and $L(\beta_1,\epsilon) < n$ for all $\epsilon \in (\epsilon_1, 1)$.
\end{proof}

\section*{Declaration of Competing Interest}
The authors declare that they have no known competing financial interests or personal relationships that could have appeared to influence the work reported in this paper.

\section*{CRediT Authorship Contribution Statement}
{\bf Ruturaj G. Gavaskar:} Formal analysis, Methodology, Writing - original draft, Writing - review \& editing, Software, Validation. {\bf Chirayu D. Athalye:} Formal analysis, Methodology, Writing - original draft, Writing - review \& editing, Supervision. {\bf Kunal N. Chaudhury:} Conceptualization, Methodology, Formal analysis, Writing - review \& editing, Supervision, Project administration, Funding acquisition.


\section*{Acknowledgments}
The work of Chirayu D. Athalye was supported by Department of Science and Technology, Government of India under Grant IFA17-ENG227.
The work of Kunal N. Chaudhury was supported by Core Research Grant CRG/2020/000527 and SERB-STAR Award STR/2021/000011 from the Department of Science and Technology, Government of India.

\bibliographystyle{elsarticle-num}
\bibliography{citations}

\end{document}